\pdfoutput=1
\documentclass[final]{siamart190516}
\usepackage{ifthen}
\usepackage{siunitx}

\newboolean{useSISC}
\setboolean{useSISC}{true}

\newboolean{presentAllData}
\setboolean{presentAllData}{false}

\usepackage[utf8]{inputenc}
\usepackage{amssymb,amsmath}

\usepackage{algorithm}
\usepackage{algpseudocode}

\usepackage{listings}

\usepackage[sort]{cite}

\usepackage{multirow}
\newcommand{\STAB}[1]{\begin{tabular}{@{}c@{}}#1\end{tabular}}

\usepackage{hyperref}

\usepackage{graphicx}
\graphicspath{ {./images/} }

\usepackage[final]{changes}

\definechangesauthor[name=Review1, color=blue]{R1}
\definechangesauthor[name=Review2, color=olive]{R2}
\definechangesauthor[name=Review3, color=red]{R3}
\definechangesauthor[name=Us, color=gray]{Add}

\newcounter{implementation}
\newenvironment{implementation}[1][]{
  \refstepcounter{implementation}\par\medskip
  \emph{Implementation remark~\theimplementation. #1} \rmfamily
}
{$\square$ \medskip}

\newcommand{\TheTitle}{A multiresolution Discrete Element Method for
triangulated objects with implicit time stepping}

\title{\TheTitle}

\ifthenelse{ \boolean{useSISC} }
{
 \title{
  \TheTitle
  \thanks{Submitted to the editors DATE.
   \funding{
    The work was funded by an EPSRC DTA PhD scholarship (award no. 1764342).
   } 
  }
 }
 \author{
   Peter J.~Noble 
    \thanks{
     Department of Computer Science, 
     Durham University
     (\email{peter.j.noble@durham.ac.uk}).
    }
   \and 
   Tobias Weinzierl
    \thanks{
     Department of Computer Science, 
     Durham University
     (\email{tobias.weinzierl@durham.ac.uk}).
    }
 }
 \headers{A Multiresolution Discrete Element Method}{Peter J.~Noble and Tobias Weinzierl}
}{
 \usepackage{hyperref}
 \usepackage{graphicx}
 \usepackage{geometry}
 \geometry{left=4.1cm,textwidth=14.2cm,top=2cm,textheight=24cm}

 \title{
  \TheTitle
 }
 \author{Peter Noble \and Tobias Weinzierl}
 \newtheorem{definition}{Definition}
 \newtheorem{corollary}{Corollary}
 \newtheorem{lemma}{Lemma}
 \newtheorem{proof}{Proof}
}

\newtheorem{observation}{Observation}
\newtheorem{assumption}{Assumption}

\DeclareMathOperator*{\argmin}{arg\,min}

\begin{document}

  \algrenewcommand\algorithmicindent{1.0em}%

  \maketitle

  \begin{abstract}
    Simulations of many rigid bodies colliding with each other sometimes yield
particularly interesting results if the colliding objects differ significantly in size and
are non-spherical.
The most expensive part within such a simulation code is the collision detection.
We propose a family of novel multiscale collision
detection algorithms that can be applied to triangulated objects within
explicit and implicit time stepping methods.
They are
well-suited to handle objects that cannot be
represented by analytical shapes or assemblies of analytical objects.
Inspired by multigrid methods and adaptive mesh refinement, we determine
collision points iteratively over a resolution hierarchy, and combine a
functional minimisation plus penalty parameters with the actual comparision-based
geometric distance calculation.
Coarse surrogate geometry representations identify ``no collision'' scenarios
early on and otherwise yield an educated guess which triangle
subsets of the next finer level \replaced[id=Add]{might}{potentially} yield
collisions.
They prune the search tree\deleted[id=Add]{,} and furthermore feed conservative
contact force estimates into the iterative solve behind an implicit time
stepping.
Implicit time stepping and non-analytical shapes often yield prohibitive high
compute cost for rigid body simulations. 
Our approach reduces \replaced[id=R1]{the object-object comparison cost}{these
cost} algorithmically by one to two orders of magnitude.
It also exhibits high
vectorisation efficiency due to its iterative nature.

  \end{abstract}
  \ifthenelse{ \boolean{useSISC} }{


  }{
  }
  
  \section{Introduction}
\label{section:introduction}

%
%
The simulation of rigid \added[id=Add]{or incompressible} bodies
\deleted[id=Add]{or rigid body parts} is a challenge that arises in many fields.
Notably,
it is at the heart of Discrete Element
\replaced[id=Add]{Method (DEM) simulations,}{simulations (DEM),} where millions
of these objects are studied.
Progress on
simulations with rigid, impenetrable objects depends
on whether we can handle high geometric detail \added[id=Add]{accurately}:
For the analysis of particle flow such as powder, it is mandatory to simulate
billions of particles\added[id=Add]{, i.e.~incompressible bodies}, while the
realism of some simulations hinges on the ability to handle
particles of different shapes and sizes
\cite{Alonso:09:EfficientDEM,Iglberger:2010:GranularFlowNonSpherical,Krestenitis:15:FastDEM,Krestenitis:17:FastDEM,Li:1998:HierarchicalBoundingVolumes,Zhong2016, Weinhart2016InfluenceOC, alonsomarroquin2008efficient, grain3d, Zhao2019APA}.
It is the support of different shapes and
sizes that allows us to simulate complex mixture
\replaced[id=Add]{phenomena, separation of scales,}{and separation of scales
phenomena,} or blockage if many objects, aka particles, try to squeeze through narrow geometries, e.g.

\begin{figure}[h]
 \begin{center}
  \includegraphics[width=\textwidth]{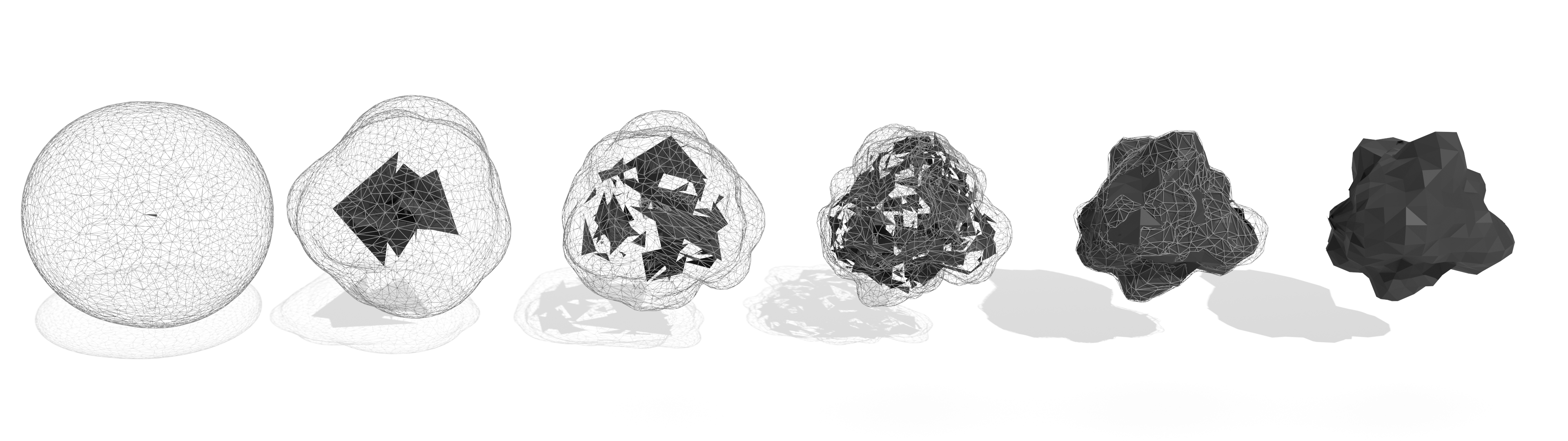}
 \end{center}
 \vspace{-0.8cm}
 \caption{
  Surrogate triangle hierarchy for a ``bumped sphere''.
  The actual sphere geometry $\mathbb{T}_h$ is shown on the right.
  \replaced[id=R3]{From left to right:
  Finer and finer}{From right to left: 
  Coarser and coarser} surrogate representations incl.~their $\epsilon
  $-environments.
  The coarser a representation, the lower the triangle count and the larger the
  corresponding $\epsilon $.
  The surrogate triangles are weakly connected.
  \added[id=R3]{The right most geometry is the actual particle with a closed surface.}
  \label{figure:multiresolution:weakly-connected}
 }
\end{figure}

%
%
DEM codes spend most of their runtime on
collision detection
\cite{Iglberger:2010:GranularFlowNonSpherical,Krestenitis:17:FastDEM,Li:1998:HierarchicalBoundingVolumes,Rakotonirina:2018:ConvexShapes}\deleted[id=Add]{---even
if they restrict to analytical shapes only}.
We tighten the challenge and study a DEM prototype over particles where 
(i) rigid particles have massively differing size, 
(ii) rigid particles are discretised by many triangles, and
(iii) rigid particles have complicated, non-convex shapes.
Our code supplements each rigid particle with an $\epsilon $-area
\cite{Alonso:09:EfficientDEM, alonsomarroquin2008efficient, grain3d, Zhao2019APA}
and considers two particles
to be ``in contact'' if their $\epsilon $-environment overlaps.
This yields a weak compressibility model, where the contact points are unique
up to an $\epsilon $-displacement.
\replaced[id=Add]{Yet, both}{Both} the arrangement and the topology of these
points \added[id=Add]{still} can change significantly between any two time
steps, while the collision models using the
contact data \replaced[id=Add]{remain}{are} inherently stiff.
\added[id=R1]{
Large-scale DEM codes require the comparison of many particles per time
step.
Efficient many-body simulations employ techniques such
as neighbour lists, cell or tree meta data structures
}
\cite{EricsonRTCD,Krestenitis:17:FastDEM, alonsomarroquin2008efficient,grain3d} 
\added[id=R1]{
to narrow down the potential collisions, i.e.~to identify particle pairs which might collide as they are spatially close.
After this pre-processing or filtering step, they detect collision points
between particle pairs.
We focus exclusively on the particle-particle comparison challenge, as 
efficient algorithms for the neighbourhood identification---including techniques
for challenging shapes and massively differing sizes---are known.
}

Our work proposes a multiscale contact detection scheme which brings down the
compute time for the contact detection aggressively.
We can handle complex, non-convex shapes and even speed up implicit
time stepping significantly.
The latter is, so far, prohibitively expensive for most codes.
To reduce the runtime, our approach phrases the contact search
\added[id=Add]{between two particles} as an iterative algorithm over multiple
resolutions where coarser particle resolutions act as surrogates.
This idea enables us to introduce five algorithmic optimisations:
First, the surrogates
help us to identify ``no contact'' constellations quickly.
\added[id=R1]{
 In this case, we can immediately terminate the search algorithm.
 The surrogates yield a generalisation of classic bounding sphere checks---if two
 bounding spheres do not overlap, the underlying objects can not overlap---to
 highly non-spherical geometries.
} Second, \replaced[id=Add]{we}{they} exploit that we do not
 represent surrogate resolutions as a plain level of detail \cite{RealTimeRendering} but make them form
a tree:
If we identify a potential collision, only those sections of the geometry are
``up-pixeled'' from where a collision point might arise from.
We increase the resolution locally.
\added[id=R1]{
 The complexity per iterative step thus does not grow exponentially as we
 switch to more accurate geometric representations.
 Instead we tend to have a linear increase of cost as we use finer and finer
 geometric models.
}
Third, the surrogate resolutions yield conservative estimates of the force that might result from
a contact point.
Once we employ an implicit time stepping scheme with a Picard iteration, 
we can permute the iterative
solver loop and the resolution switches such that the 
Picard iteration forms the outer loop which zaps through resolution
levels upon demand.
\added[id=R1]{
 The cheap surrogates provide an educated guess to the iterative force
 calculation and thus accelerate the convergence.
}
Fourth,  we phrase the contact detection as a distance minimisation
problem \cite{Krestenitis:15:FastDEM,Krestenitis:17:FastDEM}.
The minimisation problem is solved iteratively through an additional, embedded
Newton which approximates the Jacobian via a
diagonal matrix and runs through a prescribed number of sweeps.
Once more, we permute the loop over triangle pairs and \added[id=Add]{the}
Picard iterations to improve the vectorisation suitability.
Finally, we acknowledge that an iterative minimisation subject to a prescribed iteration count
can fail if the underlying geometric problem is ill-posed. 
In such cases, we eventually postprocess it by falling back to a comparison-based
distance calculation.
However, this is only required on the finest mesh, whereas we use non-convergence within the surrogate tree as 
sole ``refinement criterion'' (use a finer mesh) that does not
\replaced[id=Add]{feed forces into}{contribute forces towards} the implicit time
stepping.

To the best of our knowledge, our rigorous multiscale idea, which can be read
as a combination of (i) loop permutation and fusion, (ii) adaptive mesh refinement, (iii)
a generalisation of volume bounding hierarchies
\cite{Barequet,Dammertz,Eisenacher,Gottschalk,Held1996RealtimeCD} and (iv) an
approximate, weak closest-triangle formulation 
\cite{Krestenitis:15:FastDEM,Krestenitis:17:FastDEM}, is unprecedented.
Its reduction of computational cost plus its excellent vectorisation character
in combination with the fact that rigid body simulations 
scale well by construction---the extremely short-range
interactions fit well to domain decomposition---brings implicit DEM simulations for triangulated, non-convex
shapes within practical applications into reach.
The core algorithmic ideas furthermore have an impact well beyond
the realm of DEM.
The search for nearest neighbours, i.e.~contact within a certain environment,
is, for example, also a core challenge behind
fluid-structure interaction (FSI).

%
%
The manuscript is organised as follows:
We first introduce our algorithmic challenge and a textbook implementation of
both explicit and implicit time stepping for it (Section~\ref{section:algorithmic-framework}).
An efficient contact detection between triangulated
surfaces of two particles via a minimisation problem is introduced next (Section
\ref{section:iterative}), before
we rewrite the underlying geometric problem as a multiresolution challenge and
introduce our notion of a surrogate
data structure (Section~\ref{section:multiresolution-model}).
In Section~\ref{section:multiresolution-contact-detection}, we bring both the efficient triangle comparisons and the
tree idea together as we plug them into an implicit time stepping
code, before we allow the non-linear equation system solvers' iterations
to move up and down within the surrogate tree.
This is the core contribution of the manuscript. 
Following the discussion of some numerical results
(Section~\ref{section:results}), we sketch future
work and close the discussion.

  \section{Algorithmic framework}
\label{section:algorithmic-framework}

We study a system of $| \mathbb{P} |$ rigid bodies (particles).
Each particle $p \in \mathbb{P}$ has \deleted[id=Add]{a position $x(p,t)$,} a velocity
$v(p,t)$ and \replaced[id=Add]{an angular velocity $\omega(p,t)$ which
determine its change of position and rotation.}{a rotation $r(p,t)$.} 
Each is described by a triangular tessellation
$\mathbb{T}(p,t)$.
$t \geq 0$ is the simulation time.
We may assume that $\mathbb{T}(p,t)$ spans a well-defined, closed surface
represented by a \replaced[id=Add]{mesh where no}{conformal mesh:
No} two triangles intersect, \deleted[id=Add]{two triangles share at
most one complete edge or exactly one vertex,} and we can ``run around'' a
particle infinitely often without falling into a gap.
While the triangulation of the object is time-invariant, it moves and rotates
over time and therefore depends on $t$.

\begin{algorithm}[htb]
  \caption{
    High-level pseudo code for an explicit Euler for rigid particles.
    The continues properties \replaced[id=Add]{$v(p,t), \omega(p,t)$}{$v(p,t), r(p,t)$} and $\mathbb{T}(p,t)$ are
    discretised in time and thus become \replaced[id=Add]{$v(p), \omega(p)$}{$v(p), r(p)$} and $\mathbb{T}(p)$.
    $t$ is the (discretised) time, $\Delta t$ the time step size.
    \label{algorithm:algorithmic-framework:explicit-Euler}  
  }
 {\footnotesize
  \begin{algorithmic}[1]
    \While{$t<T_{\mbox{terminal}}$}
      \Comment{We simulate over a time span}
      \State $\forall p_i \in \mathbb{P}: \ \mathbb{C}(p_i) \gets \emptyset$ 
        \Comment{Clear set of collisions for particle $p_i$}
      \For{$p_i,p_j \in \mathbb{P},\ p_i \not= p_j$}
        \Comment{Run over all particle pairs}
        \State $\mathbb{C}(p_i) \gets \mathbb{C}(p_i) \cup $
        \Call{findContacts}{$\mathbb{T}(p_i),\mathbb{T}(p_j$)}
        \State $\mathbb{C}(p_j) \gets \mathbb{C}(p_j) \cup $
        \Call{findContacts}{$\mathbb{T}(p_i),\mathbb{T}(p_j)$}
      \EndFor
      \For{$p_i \in \mathbb{P}$}
        \State $\mathbb{T}(p_i) \gets $
        \Call{update}{$\mathbb{T}(p_i),v(p_i),\omega(p_i),\Delta t$}
        \Comment{Update geometry}
      \EndFor
        \Comment{using velocity, rotation and time step size}
      \For{$p_i \in \mathbb{P}$}
        \State $(dv,d\omega) \gets$ \Call{calcForces}{$\mathbb{C}(p_i)$}
        \State $(v,\omega)(p_i) \gets (v,\omega)(p_i) + \Delta t \cdot (dv,d\omega)$
        \Comment{Update velocity and rotation}
      \EndFor
      \State $t \gets t + \Delta t$
    \EndWhile
  \end{algorithmic}
  }
\end{algorithm}

\begin{algorithm}[htb]
  \caption{
    Contact identification between two particles $p_i$ and $p_j$.
    \label{algorithm:algorithmic-framework:contacts}  
  }
 {\footnotesize
  \begin{algorithmic}[1]
    \Function{findContacts}{$\mathbb{T}(p_i),\mathbb{T}(p_j)$}
     \State $\mathbb{C} = \emptyset$
        \For{$t_i \in \mathbb{T}(p_i), t_j \in \mathbb{T}(p_j), \ t_i \not= t_j$}
        \Comment{Run over all triangles pairs}
          \State $c \gets \textsc{contact}(t_i,t_j)$
            \Comment{Find closest point in-between $t_i$ and $t_j$ and compare normal} 
          \If{$c \not= \bot$}
            \Comment{$|n|$ against $\epsilon$; return
              $\bot $  if $|n| > \epsilon$} 
            \State $\mathbb{C} \gets \mathbb{C} \cup \{c\}$ 
          \EndIf
        \EndFor
     \State \Return $\mathbb{C}$
    \EndFunction
  \end{algorithmic}
  }
\end{algorithm}

\paragraph{Time stepping}

%
%
A straightforward high-level implementation of an explicit Euler for DEM
consists of a time loop hosting a sequence of further loops (Algorithm
\ref{algorithm:algorithmic-framework:explicit-Euler}):
The first inner loop identifies all contact points between the particles.
Once we know all contact points per particle, we can determine a velocity 
and rotation update ($dv(p)$ and \replaced[id=Add]{$d\omega(p)$}{$dr(p)$}) per
particle.
Before we do so, we update the particles' positions using their velocity and
angular momentum.
Finally, we progress in time.
Euler-Cromer would result from a permutation of the update sequence.

%
%
It is impossible to simulate exact incompressibility with explicit time stepping
schemes
\added[id=Add]{
when no interpenetration is allowed}: 
Everytime we update a particle, we run 
risk that it slightly penetrates another
one due to the finite time step size $\Delta t$.
At the same time, particles exchange no momentum as long as they are not in
direct contact yet.
The momentum exchange remains ``trivial'' until we have violated the
rigid body constraint.
For these two reasons, we switch to a weak incompressibility
model where each particle is surrounded by an $\epsilon >0$ area
\cite{Alonso:09:EfficientDEM, Zhao2019APA, grain3d,
alonsomarroquin2008efficient}.
The area is spanned by the Minkovski sum of the triangles from
$\mathbb{T}(p,t)$ and a sphere of radius $\epsilon $ minus the actual rigid object.
Without loss of generality, we assume a uniform $\epsilon $ per particle.
Our formalism is equivalent to a soft boundary formulation with an
extrusive surface.

\begin{definition}[Contact point]
 \label{definition:contact-point}
 Each particle is surrounded by an $\epsilon $-environment.
 Two particles are in contact, if their $\epsilon $-environment overlaps. 
 Overlaps yield contact points which in turn yield forces.
\end{definition}

\noindent
We parameterise each contact point with its penetration depth:
Our contact detection in Algorithm
\ref{algorithm:algorithmic-framework:contacts} identifies the closest path
between two triangles.
A potential contact point $c$ is located at the centre of
this line.
It \deleted[id=R3]{furthermore} is equipped with a normal $n(c)$, which points \replaced[id=R3]{from the contact point}{along the line}
towards either of the closest triangles (Figure~\ref{figure:multiresolution-model:contact-point}). 
There is an overlap between the two $\epsilon $-augmented triangles if and only
if $|n(c)|\leq \epsilon$\added[id=R3]{, that is if the $\epsilon $-environments penetrate}.
\replaced[id=R3]{In this case}{Such a point} $c$ is added to the set of collision points.

With a set of contact points plus their normals \added[id=R3]{and dimensionless masses $M(p_i)$ and $M(p_j)$}, 
we can derive \replaced[id=Add]{the}{a} force acting on a particle\deleted[id=R3]{ and, taking
the centre of mass and the mass of a particle into account, determine its acceleration and torque }.

\[
  F(c) = \frac{n(c)}{|n(c)|} K_s (1 - \frac{|n(c)|}{\epsilon})
  \sqrt{\frac{1}{\frac{1}{M(p_i)} + \frac{1}{M(p_j)}}}
\]

\noindent
\replaced[id=Add]{computes the force arising from one contact point by mapping
the $\epsilon$-area onto}{We compute forces by modelling the $\epsilon$-area as} 
a simple spring with a perpendicular friction force \added[id=R3]{yet without
any empirical damping} \cite{Cundall}.
\added[id=R1]{
 The force depends on the contact normal $n(c)$ and applies to both colliding
 particles} \added[id=R3]{$p_i$ and $p_j$} 
\added[id=R1]{
 subject to a minus sign for one of
 them.
 It is calibrated by a spring constant $K_s = 1,000$}
\deleted[id=R3]{ and the
 mass of the two involved particles $M(p_i)$ and $M(p_j)$}.
\added[id=R1]{
 A particle's total force is}
\added[id=R3]{then} 
\added[id=R1]{ the sum over the
 individual contact forces.
}
\added[id=R3]{
Taking the centre of mass, the total mass and the mass distribution of a
particle into account, the total force and total torque determine its
acceleration and angular acceleration}  \cite{Baraff97anintroduction}.

This is a simplistic presentation---we ignore for example sophisticated interaction functions
which distinguish contact points from contact faces---\replaced[id=R1]{
 where the discontinuous force that arises at a contact is approximated by a
 ``fade-in'' force as two particles approach each other:
 Over an interval of size $\epsilon $, the force smoothly approximates the
 target force, as the penetration depth $|n(c)|$ increases.
 The simple physics allow}
{yet it allows} 
us to focus on the core challenge how to find contact points
efficiently.

\begin{implementation}
 As we work with triangulated objects and derive contact points from
 triangle-triangle comparisons, the algorithm identifies some contact points
 redundantly:
 If the closest distance between two objects is spanned by two object
 vertices $x_i$ and $x_j$, every triangle combination $t_i$ and $t_j$ where
 $t_i$ is adjacent to $x_i$ and $t_j$ to $x_j$ finds the same
 contact point and consequently adds it to $\mathbb{C}$ (Algorithm
 \ref{algorithm:algorithmic-framework:contacts}). The set notation for
 $\mathbb{C}$ highlights that we do not work with redundant contact points:
 In the implementation, we run over $\mathbb{C}$ and merge close-by contact
 points, i.e.~points closer than $\epsilon$, into one average point.
 This filter step prior to any use of the elements from $\mathbb{C}$ is
 necessary, as we work with floating point arithmetics and an 
 $\epsilon > 0$, i.e.~a contact point is not a unique point in
 space and might temporarily exist multiple times within $\mathbb{C}$ with
 slightly different coordinates.
\end{implementation}

\begin{figure}[h]
 \begin{center}
  \includegraphics[width=0.70\textwidth]{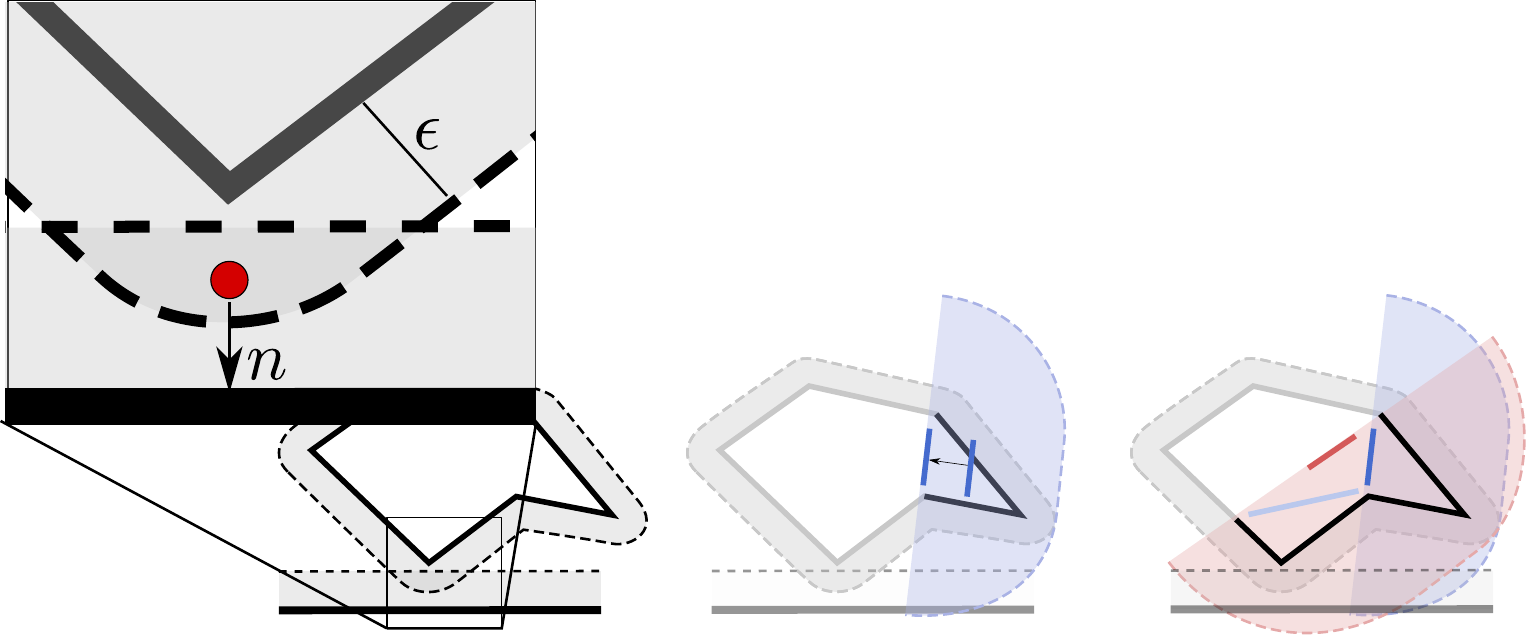}
 \end{center}
 \caption{
  A pair of objects (black, solid) with their $\epsilon $-environment (black,
  dotted) collide.
  In the present sketch, one object is a ``spherical'' particle spanned by six
  edges, while the other object is a plane at the bottom.
  Left:
  Two-dimensional sketch of the contact point concept.
  The zoom-in shows the contact point which is located in the middle of
  the overlap region.
  The contact normal is directed from the middle point towards the closest point
  on one of the objects involved.
  Middle:
  When an object hosts very extruded features, we slightly shrink the surrogate
  such that the surrogate without $\epsilon $ becomes 
  a closer fit around the "un-bumped" real geometry.
  We trade such a surrogate for a bigger $\epsilon $.
  Empirical evidence suggests that this yields slightly advantageous forces
  within a multiscale iterative solve.
  Right:
  The conservative property of the surrogate triangles states
  that all fine level geometry (including its epsilon boundary) must be encompassed by the surrogate's epsilon.
  This doesn't suggest that all surrogate children are included.
  \label{figure:multiresolution-model:contact-point}
 }
\end{figure}

\begin{algorithm}[htb]
 \caption{
  High-level pseudo code for an implicit Euler.
  \label{algorithm:algorithmic-framework:implicit-Euler}  
 }
 {\footnotesize
 \begin{algorithmic}[1]
  \While{$t<T_{\mbox{terminal}}$}
    \Comment{We simulate over a time span}
   \State $\forall p_i \in \mathbb{P}: \ 
     \mathbb{T}^{\text{guess}}(p_i) \gets \mathbb{T}(p_i),\ 
     v^{\text{guess}}(p_i) \gets v(p_i),\ 
     \omega^{\text{guess}}(p_i) \gets \omega(p_i) 
    $
   \While{$\mathbb{T}^{\text{guess}}, v^{\text{guess}}$ or $\omega^{\text{guess}}$
    change significantly for any $p$} 
   \State $\forall p_i \in \mathbb{P}: \ \mathbb{C}(p_i) \gets
   \emptyset$
     \Comment{Clear set of collisions for particle $p_i$}
    \For{$p_i,p_j \in \mathbb{P},\ p_i \not= p_j$}
     \Comment{Run over all particle pairs}
     \State $\mathbb{C}(p_i) \gets\mathbb{C}(p_i) \cup $
      \Call{findContacts}{$\mathbb{T}^{\text{guess}}(p_i),\mathbb{T}^{\text{guess}}(p_j)$}
     \State
      $\mathbb{C}(p_j) \gets\mathbb{C}(p_j) \cup $
      \Call{findContacts}{$\mathbb{T}^{\text{guess}}(p_i),\mathbb{T}^{\text{guess}}(p_j)$}
    \EndFor
    \For{$p_i \in \mathbb{P}$}
     \State $(dv,d\omega) \gets$ \Call{calcForces}{$\mathbb{C}(p_i)$}
     \State $(v^{\text{guess}},\omega^{\text{guess}})(p_i) \gets (v,\omega)(p_i) +
      \Delta t \cdot (dv,d\omega)$ 
     \State $\mathbb{T}^{\text{guess}}(p_i) \gets $
     \Call{update}{$\mathbb{T}(p_i),v^{\text{guess}}(p_i),\omega^{\text{guess}}(p_i),\Delta t$}
    \EndFor
   \EndWhile
   \State $\forall p_i \in \mathbb{P}: \ 
     \mathbb{T}(p_i) \gets \mathbb{T}^{\text{guess}}(p_i),\ 
     v(p_i) \gets v^{\text{guess}}(p_i),\ 
     \omega(p_i) \gets \omega^{\text{guess}}(p_i) 
    $
   \State $t \gets t + \Delta t$
  \EndWhile
 \end{algorithmic}
 }
\end{algorithm}

\paragraph{Implicit time stepping}

An implicit Euler for \replaced[id=Add]{a DEM code}{this challenge} has to solve
a non-trivial, non-linear equation system per time step.
Non-trivial means that the system's sparsity
pattern depends on the solution of the system itself.
It is determined by the contact point search: we obtain entries
in the interaction matrix where the corresponding normal $|n|\leq \epsilon$.
Non-linearity means that the quantities in
the interaction matrix (forces) depend on the (guess of the) geometries'
arrangement.

\begin{assumption}
 \label{assumption:contraction}
 Our implicit time stepping problem exhibits a contraction property, i.e.~Picard iterations 
 can solve the underlying non-linear equation system.
\end{assumption}

\noindent
As our non-linear system is ``well-behaved''---we employ reasonably small
$\Delta t$---we rely on a fixed-point
formulation of the implicit \replaced[id=Add]{time stepping}{timestepping} and
employ Picard iterations, i.e.~we approximate the velocity and 
\replaced[id=Add]{angular velocity}{rotation}
\replaced[id=Add]{$(dv,d\omega)(p)(t)$}{$(dv,dr)(p)(t)$} through a repeated
application of the contact detection plus its following force calculation (Algorithm \ref{algorithm:algorithmic-framework:implicit-Euler}).
We assume that we remain within the Picard iterations' region of convergence.

Picard avoids the assembly and inversion of an equation system.
However, many contact points enter the
algorithm temporarily throughout the iterations, which eventually are not identified as contacts.
This happens, for example, if we rotate the particles ``too far'' throughout the iterations.
Despite small $\Delta t$, 
we cannot provide an upper bound on the number of 
Picard iterations required or make
assumptions on the (temporarily) identified contact points, i.e.~the cost per iteration.

\paragraph{Notation and terminology}

Any particle $p \in \mathbb{P}$ has a (closed) volume $V(p)$
which is spanned by its triangular surface $\mathbb{T}(p)$.
Since we stick to explicit and implicit single-step, single shot methods, we
omit the parameterisation $(t)$ from hereon.
\[
  p_1 \cap p_2 = \emptyset \qquad \forall p_1, p_2 \in \mathbb{P} \ \text{with}
  \ p_1 \not= p_2,
  \label{equation:framework:non-overlapping-particles}
\]

\noindent
as we have rigid, non-penetrating objects.
Consequently, $ \mathbb{T}(p_1) \cap \mathbb{T}(p_2) = \emptyset $.
The particles' triangles do not intersect.
Yet, as each particle is surrounded by an $\epsilon $-layer, our particles'
triangles $t \in \mathbb{T}(p)$ unfold into a set of volumetric objects
$t^{\epsilon } \in \mathbb{T}^{\epsilon }(p) $, 
and our particles' volumes are extended $p \subset p^{\epsilon }$, too.
Overlaps between extended volumes do exist and yield contact
points:
\begin{equation}
  p_1^{\epsilon } \cap p_2^{\epsilon } \not= \emptyset \ \Rightarrow
  \mbox{contact point}
  \label{equation:framework:contact-point}
\end{equation}

\noindent
which is equivalent to
\[
  \exists t_1 \in \mathbb{T}(p_1),
          t_2 \in \mathbb{T}(p_2):
          \quad
  t_1^{\epsilon } \cap t_2^{\epsilon } \not= \emptyset \ \Rightarrow
  \mbox{contact point}.
\]

\noindent
A contact point $c$ between two
triangles $t_i$ and $t_j$ in 
(\ref{equation:framework:contact-point}) is located at 
\begin{equation}
  x(c) = \frac{1}{2}(t_1(a_1,b_1) + t_2(a_2,b_2)) 
  \ \text{with} \
  \argmin _{a_1,b_1,a_2,b_2 \in [0,1] } \frac{1}{2} \| t_1(a_1,b_1) -
  t_2(a_2,b_2) \|^2,
  \label{equation:framework:minimisation}
\end{equation}

\noindent
if $t_1$ and $t_2$
are surrounded by the same $\epsilon $ halo.
For different $\epsilon $s, the weights for $x(c)$ have to be
adapted accordingly.
$a_1,b_2$ are Barycentric coordinates within $t_1$, i.e.~$t_1(a_1,b_1)$ returns
a position within the triangle.
$a_2$ and $b_2$ are the counterparts for $t_2$.
If $x(c)$ is the position of the contact point $c$, the corresponding normal 
$ n(c) = t_1(a_1,b_1) - x(c)$ or $n(c) = t_2(a_2,b_2) - x(c)$.

  \section{Iterative contact detection via distance minimisation}
\label{section:iterative}

%
%
To find the closest point between two triangles is a classic
computational geometry problem \cite{EricsonRTCD}.  
We rely on three different algorithms to \replaced[id=Add]{solve}{identify} it:

\paragraph{Direct distance calculation (comparison-based)} 

A comparison-based identification of contact points consists of
two steps.
First, we compute
the distance from each vertex of the two triangles to the
closest point on the other triangle as well as the distance between each pair of edges
between the two triangles \cite{EricsonRTCD}.
This yields six point-to-triangle distance tests and
nine edge-to-edge distance tests.
In a second step, we select the minimum distance out of the 15 combinations.
This \replaced[id=Add]{brute force}{\emph{brute force}} calculation yields an
exact solution---agnostic of truncation errors---yet requires up to 30+14 comparisons (if statements) per triangle pair,
which we tune via masking, blending and early termination \cite{SSETriangle}.
\added[id=Add]{Dealing with the cases where intersections between triangles is possible
requires an additional six edge-to-plane distance tests,
where intersections outside the area of the triangle are discounted.}

\paragraph{Iterative distance calculation}

As an alternative to a comparison-based approach, we replace 
the geometric checks with a functional where 
we minimise the distance between the two planes spanned by the triangles but
add the admissibility conditions over the Barycentric coordinates as Lagrangian parameters
\cite{Krestenitis:15:FastDEM,Krestenitis:17:FastDEM}.
In line with (\ref{equation:framework:minimisation}), let $a,b \in [0,1]$  
describe any point on their respective triangle:
\begin{eqnarray}
 \argmin _{a_1,b_1,a_2,b_2} J(a_1,b_1,a_2,b_2) :=
 &&
 \argmin _{a_1,b_1,a_2,b_2} 
 \underbrace{
 \frac{1}{2} \| t_i(a_1,b_1) - t_j(a_2,b_2) \| ^2
 }_{=:\hat J(a_1,b_1,a_2,b_2)}
 + \alpha _{\text{iterative}} \Big(   
 \nonumber \\  
 &&\max (0,a_1-1) + \min(-a_1,0) + \max (0,b_1-1) + \nonumber\\
 &&\max (-b_1,0)  + \max (0,a_1+b_1-1) + \nonumber\\
 &&\max (0,a_2-1) + \min(-a_2,0) + \max (0,b_2-1) + \nonumber\\
 &&\max (-b_2,0)  + \max (0,a_2+b_2-1)
 \Big).
  \label{equation:iterative-contact-detection:energy-functional}
\end{eqnarray}

\noindent
This is a weak formulation of the challenge, since \replaced[id=Add]{any $\alpha
_{\text{iterative}}<\infty $}{it} allows the closest distance line between two
triangles to be rooted slightly outside the very triangles.

The minimisation problem can be solved via Newton iterations.
However, the arising \replaced[id=R3]{Hessian}{Hamiltonian} becomes difficult to invert or non-invertible for close-to-parallel or actually parallel
triangles.
We therefore regularise it by adding a diagonal matrix.
After that, we approximate the regularised
\replaced[id=R3]{Hessian}{Hamiltonian} and update the minimisation and constraints alternatingly:
\begin{eqnarray}
 (a_1,b_1,a_2,b_2)^{(n+0.5)} & = &
 (a_1,b_1,a_2,b_2)^{(n)} -
 \text{diag} ^{-1} \Big( \hat H(a_1,b_1,a_2,b_2)^{(n)} 
 \nonumber
 \\
 && + \alpha _{\text{regulariser}} \ id \Big) 
 \nabla_{a_1,b_1,a_2,b_2} \hat J(a_1,b_1,a_2,b_2)^{(n)}  
 \nonumber
 \\
 (a_1,b_1,a_2,b_2)^{(n+1)} & = &
 (a_1,b_1,a_2,b_2)^{(n+0.5)} -
 \text{diag} ^{-1} \Big( \tilde{H}(a_1,b_1,a_2,b_2)^{(n+0.5)}
 \nonumber
 \\
 && + \alpha _{\text{regulariser}} \ id \Big) 
 \nabla_{a_1,b_1,a_2,b_2} \tilde{J}(a_1,b_1,a_2,b_2)^{(n+0.5)}  
 \label{equation:iterative:modified-gradient-descent}
\end{eqnarray}

\noindent
$\alpha _{\text{regulariser}}>0$ is small, while 
$^{(n)}$ is the iteration index.
The $\hat J$ and its \replaced[id=R3]{Hessian}{Hamiltonian} correspond to the quadratic
term of the actual functional in (\ref{equation:iterative-contact-detection:energy-functional}) aka (\ref{equation:framework:minimisation}).
$\tilde{J}$ and its \replaced[id=R3]{Hessian}{Hamiltonian} cover the remaining penalty terms, i.e.~$J - \hat J$.
However, we omit the Dirac \deleted[id=R3]{jump} terms in there, i.e.~we explicitly drop terms that enter the formulae for $a_1=0,
a_1=1, b_1=0, \ldots$.
Our solver iterates back and forth between the $\hat J$-minimisation and
a fullfillment of the constraints. 
This modified Newton becomes a gradient descent, where \replaced[id=Add]{the}{an
approximate} step size is adaptively chosen by analysing an approximation of the
inverse to the \replaced[id=R3]{Hessian}{Hamiltonian}.

\paragraph{Hybrid distance calculation}

If two subsequent iterates $|J(a_1,b_1,a_2,b_2)^{(n+1)} -
J(a_1,b_1,a_2,b_2)^{(n)}| \leq C\epsilon$\added[id=Add]{,} we have found an
actual minimum over functional (\ref{equation:iterative-contact-detection:energy-functional}) and can terminate the minimisation.
In this case, we assume that $a_1,b_1,a_2,b_2$ identify the minimum distance. 
Without further analysis, it is impossible to make a statement on the upper bound on
$n$.
Our hybrid algorithm therefore eliminates the termination criterion and imposes $n \leq N_{\texttt{iterative}}$.
Consequently, it labels a distance calculation as
invalid if $|J(a_1,b_1,a_2,b_2)^{(N_{\texttt{iterative}})} - J(a_1,b_1,a_2,b_2)^{(N_{\texttt{iterative}}-1)}| > C\epsilon $.
The iterative code's result realises a three-valued logic:
``found a contact point'', ``there is no contact'', or ``has not
\replaced[id=Add]{terminated}{termianted}'' (Algorithm
\ref{algorithm:triangle-soup}).

Our hybrid algorithm \deleted[id=Add]{relies upon this iterative building block.
It first} invokes the modified iterative algorithm.
If the result equals ``not terminated'' ($\odot $), the hybrid algorithm falls back to the 
comparison-based distance calculation.
It is thus not really a third algorithm to find a contact point, but a 
combination of the iterative scheme with the comparison-based approach
serving as a posteriori limiter.

\begin{algorithm}[htb]
 \caption{
  A \added[id=Add]{hybrid, batched} reformulation of the iterative distance
  calculation.
  The
  iterative \replaced[id=Add]{sweep}{blueprint (top)} over a whole batch of
  triangles \replaced[id=Add]{is stripped of a dynamic stopping criterion (top)
  and therefore yields three types of results per triangle pair:}{and early
  convergence stopping criteria is removed to rewrite it into a version that vectorises more aggressively (bottom).}
  $x$ \deleted[id=Add]{is a result value that} holds a coordinate if the
  Barycentric coordinates yield a contact point, or $\bot$ if they yield no
  contact point, or $\odot $ if the triangle combination has to be postprocessed with the comparison-based algorithm\added[id=Add]{ (bottom)}.
  \label{algorithm:triangle-soup}  
 }
 {\footnotesize
 \begin{algorithmic}[1]
  \For{$t_i \in \mathbb{T}(p_i), \ t_j \in \mathbb{T}(p_j), \ t_i \not= t_j$}
    \State $(a_1,b_1,a_2,b_2,n)(t_i,t_j) \gets 0$
     \Comment Variables per triangle pair allow us to vectorise 
    \State $J_\text{old}(t_i,t_j)  \gets \infty$
     \Comment over triangles in $\mathbb{T}$ 
    \State $J(t_i,t_j)=J(a_1(t_i,t_j),b_1(t_i,t_j),a_2(t_i,t_j),b_2(t_i,t_j))$
     \Comment Functional from
     (\ref{equation:iterative-contact-detection:energy-functional}) 
    \While{$n\leq N_{\text{iterative}}$}
    \Comment Fixed iteration count instead of $|J-J_\text{old}| > C\epsilon $
      \State $J_\text{old}(t_i,t_j) \gets J(t_i,t_j)$
      \State update $a_1(t),b_1(t),a_2(t),b_2(t)$
       \Comment Modified gradient descent from
       (\ref{equation:iterative:modified-gradient-descent}) 
      \State $J(t_i,t_j)=J(a_1(t_i,t_j),b_1(t_i,t_j),a_2(t_i,t_j),b_2(t_i,t_j))$
      \State $n \gets n+1$
    \EndWhile
   \State $x(t_i,t_j) \gets
    \left\{
     \begin{array}{cc}
      \frac{1}{2} ( t_i(a_1,b_1) - t_j(a_2,b_2) )(t_i,t_j)
        & \text{if} \ |J(t_i,t_j)-J_\text{old}(t_i,t_j)| \leq C\epsilon
      \wedge \hat J(t_i,t_j) \leq 2\epsilon ^2 \\
      \bot & \text{if} \ |J(t_i,t_j)-J_\text{old}(t_i,t_j)| \leq C\epsilon \wedge \hat J(t_i,t_j)
      > 2\epsilon ^2 \\
      \odot & \text{otherwise}
     \end{array} 
    \right.$ 
  \EndFor
  \\\hrulefill
  \For{$t_i \in \mathbb{T}(p_i), \ t_j \in \mathbb{T}(p_j), \ t_i \not= t_j$}
    \If{$x(t_i,t_j)=\odot$}
      \State $\hat n(t_i,t_j) \gets $ shortest distance vector between $t_i$ and
      $t_j$ 
      \Comment Use comparison-based algorithm
      \State $\hat x(t_i,t_j) \gets $ is central point on line identified
      by $\hat n(t_i,t_j)$ 
      \State $x(t_i,t_j) \gets
    \left\{
     \begin{array}{cc}
      \hat x(t_i,t_j) & \text{if} \ |\hat n(t_i,t_j)| \leq 2\epsilon \\
      \bot & \text{otherwise}
     \end{array} 
    \right.$ 
      \Comment Eliminate $\odot $ entries in result
    \EndIf
  \EndFor
 \end{algorithmic}
 }
\end{algorithm}

\begin{implementation}
Triangle meshes are typically held as graphs over vertex sets. 
We flatten this graph prior to the contact detection: 
A sequence of $|\mathbb{T}|$ triangles is converted into a sequence of $3
\cdot 3 \cdot |\mathbb{T}|$ floating point values, i.e.~each triangle is
represented by the coordinates (three components) of its three vertices.
Such a flat data structure can be generated once prior to the first contact
detection.
All topology is lost in this representation and vertex data is replicated---it
is a triangle soup---but the data is well-suited to be streamed without
indirect memory lookups.
Per update of the position and rotation, every coordinate is
subject to an affine transformation.

\deleted[id=R2]{
Let $\mathbb{T}_i$ and $\mathbb{T}_j$ be two triangle sets that have to be compared.
To improve the vector efficiency, we
tile one of the flattened input triangle sets, which are stored as struct of
arrays, into batches of size $N_{\text{batch}}$ (Algorithm
\ref{algorithm:triangle-soup}), such that the memory access patterns fit to the vector register
width.
We next exploit the fact that the iteration count is bounded or 
hardcoded---i.e.~we omit an early termination criterion---and consequently
permute the loop over the comparision triangles and the modified Newton
iterations.
As we compute the result for $N_{\text{batch}}$ triangle combinations in one embarassingly parallel 
rush, we can vectorise aggressively.
}
\end{implementation}


  \section{A multiresolution model}
\label{section:multiresolution-model}
The cost to compare two particles $p_i$ and $p_j$ is
determined by their triangle counts $| \mathbb{T}(p_j) |$ and $| \mathbb{T}(p_j) |$.
To reduce this cardinality,
we construct geometric cascades of triangle models
per particle---the \emph{surrogate models} (Fig.~\ref{figure:multiresolution:weakly-connected})---and plug
representations from this cascade into Algorithm \ref{algorithm:algorithmic-framework:contacts}.

\begin{definition}
 A surrogate model $\mathbb{T}_{k}(p),\ k\geq 1$ is a triangle-based geometric
 approximation of a particle described by
 $\mathbb{T}(p)$.
 A sequence of surrogate models
 $\{\mathbb{T}_{1}(p),\mathbb{T}_{2}(p),\mathbb{T}_{3}(p),\ldots\}$
 for $p$ with its volumetric extensions $\mathbb{T}_{k}^{\epsilon}(p)$ hosts efficient (Definition \ref{definition:multiresolution:efficient}),
 \emph{conservative}
 (Definition \ref{definition:multiresolution:conservative}), and \emph{weakly
 connected}
 (Definition \ref{definition:multiresolution:weakly-connected})
 abstractions of $\mathbb{T}(p)$.
\end{definition}

\noindent
Surrogate models are different representations of a particle.
 The term ``cascade'' highlights that each particle is assigned a whole sequence
 of \deleted[id=R3]{accurate} representations.
 These representations 
can step in for our real geometry and are a special class of bounding volume techniques \cite{EricsonRTCD}.
To simplify our notation, $\mathbb{T}_{0}(p) := \mathbb{T}(p)$\added[id=R1]{,
i.e.~the $k=0$ surrogate model is the geometric object itself.
The bigger $k$, the coarser, i.e.~more abstract the surrogate.
}.
\deleted[id=R1]{We}{Furthermore, we} emphasise that the $\epsilon $ is a generic
symbol, i.e.~each surrogate model hosts its own, bespoke $\epsilon$-environment.

\begin{definition}
 \label{definition:multiresolution:efficient}
 A surrogate model
 $\mathbb{T}^\epsilon_{k_i}(p_i)$ is
 \emph{efficient} if, for any other model 
 $\mathbb{T}^\epsilon_{k_j}(p_j)$,
 finding all contact points between 
 $\mathbb{T}^\epsilon_{k_i}(p_i)$
 and 
 $\mathbb{T}^\epsilon_{k_j}(p_j)$
 is cheaper than finding all contact points between 
 $\mathbb{T}^\epsilon_{\hat{k}_i}(p_i)$
 and 
 $\mathbb{T}^\epsilon_{k_j}(p_j)$
 for all $0 \leq \hat{k}_i<k_i$.
\end{definition}

\noindent
We assume an almost homogeneous cost per triangle-to-triangle
comparison---an assumption that is shaky for the hybrid
\replaced[id=Add]{comparison}{comparision} and subject to vectorisation
efficiency and memory management effects.
Hence, the triangle count of surrogate models decreases with increasing $k$,
i.e.~$|\mathbb{T}_{k}| \ll |\mathbb{T}_{k+1}|$.

\begin{definition}
 \label{definition:multiresolution:conservative}
 A surrogate model 
 $\mathbb{T}_{k_i}(p_i)$ inducing $\mathbb{T}_{k_i}^{\epsilon }(p_i)$
 is \emph{conservative} if 
 \[
  \forall p_i, p_j, k_i, k_j: \
  \mathbb{T}_{k_i}^{\epsilon }(p_i) \cap \mathbb{T}_{k_j}^{\epsilon
  }(p_j) = \emptyset \ \Rightarrow \
  \mathbb{T}^{\epsilon }(p_i) \cap \mathbb{T}^{\epsilon }(p_j)
  = \emptyset.
 \]
\end{definition}

\noindent
Conservative means that any two surrogates of two
particles are in contact (overlap) if the two particles are in contact.
Yet, \replaced[id=Add]{this does not have to hold the other way round: If}{if}
their surrogates are in contact, there might still be gaps between the particles, i.e.~there might be no contact point.

\begin{corollary}
 Let a surrogate model hierarchy $\mathbb{T}_{k}(p), \ k \geq 1$
 be \emph{monotonous} if 
 \[
  \forall 1 \leq \hat{k} < k: 
  \quad
  \forall 
  t_{\hat{k}}^{\epsilon } 
  \in
  \mathbb{T}_{\hat{k}}^{\epsilon }(p):
  \quad
  t_{\hat{k}}^{\epsilon _{\hat{k}}} 
  \subseteq
  \bigcup _{t \in \mathbb{T}_{k}^{\epsilon }(p)} t.
 \]
  Our surrogate hierarchies do not have to be monotonous.
\end{corollary}

\noindent
A monotonous cascade of triangles plus $\epsilon $-environments would
grow in space as
we move up the hierarchy of models.
Therefore, we do not impose it, even though monotonicity would imply
conservativeness ``for free''.
Empirical evidence suggests that abandoning monotonicity 
allows us to work with significantly tighter $\epsilon $-choices per level.
Yet, it also implies that we generally cannot show algorithm correctness through plain induction.

Classic level-of-detail algorithms require a coarsened representation of a
triangulated model to preserve certain properties such as connected triangle
surfaces or the preservation of certain features such as sharp edges. Our
surrogate models however are to be used as temporary replacements within our
calculations. 
We thus can ask for weak representations of the geometries:

\begin{definition}
 \label{definition:multiresolution:weakly-connected}
  The triangles of a particle have to span a connected surface.
  For surrogate models, there is no such constraint.
  Therefore, the surrogate models
  can be \emph{weakly connected}:
  Their triangles can be disjoint with gaps in-between or they can intersect each other.
  A surrogate's $\epsilon $-environment is connected however, and
  it covers (overlaps) all connectivity of the original
  model.
\end{definition}

\noindent

\noindent
The connectivity addendum in Definition
\ref{definition:multiresolution:weakly-connected} motivates the term weakly
connected.
It clarifies that we---despite the disjoint, anarchic configuration of a
surrogate triangle set (Figure~\ref{figure:multiresolution:weakly-connected})---
 cannot miss out on some geometry extrema such as
sharp edges due to tests with surrogates:
If there is no contact between two surrogates, there is also no contact between
their real discretisations.

Definition \ref{definition:multiresolution:conservative} implies that we
can use surrogate models as guards and run through them for coarse to fine:
If there are no contact points between two surrogate models,
there can be no contact points for the more detailed models.
We can stop searching for contact points immediately.
It does however not hold the other way round.
Definition \ref{definition:multiresolution:efficient} implies immediately that
the number of triangles that we examine in such an iterative approach is
monotonously growing.
Definition \ref{definition:multiresolution:weakly-connected} gives us the
freedom to construct such triangle hierarchies, as it strips us from many
geometric constraints.


\begin{definition}
 \label{definition:multiresolution-model:surrogate-tree}
  \added[id=R3]{
   Let}
    $\mathcal{T}$ \added[id=R3]{be a directed acycling graph where each node
    represents a set of triangles.
   The level} $\ell $ \added[id=R3]{of a node is its distance (edge count) from
   the root node in} $\mathcal{T}$.
   \added[id=R3]{The resulting graph is a} \emph{surrogate tree}
   \added[id=R3]{if and only if
  }
\begin{enumerate}
   \item \added[id=R3]{
 the root node hosts the coarsest surrogate model}
 $\mathbb{T}_{k_{\text{max}}}$;
   \item \added[id=R3]{the union over all leaf sets yields the particle triangulation
   } $\mathbb{T}_0$;
    \item \replaced[id=R3]{
    any triangle is a surrogate for the union over its
 children's triangle sets.
    }{
    The union over all sets with the same level yields the surrogate model
    }
\end{enumerate}    
\end{definition}

 \noindent
 \added[id=R3]{
  With Definition}
 \ref{definition:multiresolution-model:surrogate-tree}, \added[id=R3]{the union
 over all sets with the same level yields the surrogate model} $\mathbb{T}_{k_{\text{max}}-\ell}$.
 \deleted[id=R3]{
 To make the last property in Definition
 \ref{definition:multiresolution-model:surrogate-tree} hold, a triangle within a
 node of the tree which has children has to be surrogate for the union over its
 children's triangle sets.
 }
 We use $N_{\text{surrogate}}$ to denote the number of children of a surrogate
 triangle.
 $N_{\text{surrogate}}$\added[id=R1] {
  does not have to be uniform over the tree, i.e.~is a
 generic symbol.
 If we have a surrogate model, take the nodes within the tree which hold its
 triangles, and replace the triangles with those triangles stored within
 children nodes, we obtain the next finer surrogate.
 A surrogate tree is a generalisation of the concept of a cascade of surrogates:
 The tree formalism allows us also to construct different,
 hybrid-level surrogate models if we only replace some triangles of a model with their
 children.
}

\begin{implementation}
 There are multiple ways to construct surrogate trees.
 We construct our trees through a recursive algorithm.
 It starts from the triangle set $\mathbb{T} = \mathbb{T}_0$ of the particle and
 splits this set into $N_{\text{surrogate}}$ subsets of roughly
 the same size hosting close-by triangles. 
 Per subset, we construct one surrogate and thus obtain a tree of depth one
 where the root node hosts $|N_{\text{surrogate}}|$ triangles.
 As long as a node within the tree hosts more triangles than a prescribed
 threshold, we apply the splitting recursively and thus disentangle the triangle
 sets further and further:
 Existing tree levels are pushed down or sieved through the tree hierarchy (Appendix
 \ref{appendix:surrogate-trees}).
 A follow-up bottom-up traversal of the tree constructs well-suited surrogate
 triangles with appropriate $\epsilon $ choices.
 As we only require weakly connected surrogates, the steps within this bottom-up
 traversal are independent of operations on sibling nodes within the tree and
 can be mapped onto local minimisation problems (Appendix
 \ref{appendix:surrogate-triangles}).
\end{implementation}

  \section{Multiresolution contact detection}
\label{section:multiresolution-contact-detection}

With our surrogate tree definition, we are in the position to propose a
multiscale algorithm for \replaced[id=R3]{an}{the} explicit Euler which utilises the tree as early
stopping criterion\replaced[id=Add]{ for the search for contacts. The
observations for the explicit time stepping then allow us to}{, and we can,
starting from the explicit Euler observation,}  derive two implicit
time stepping algorithms that exploit the multiscale nature of the
geometry\added[id=Add]{ to terminate searches early and to supplement the
underlying iterative algorithm with educated guesses}:

\subsection{Explicit Euler}
\label{subsection:multiscale:explicit-Euler}
%
%
Our explicit Euler exploits the multiscale hierarchy by looping over the
 the resolutions held in $\mathcal{T}$
 top down.
The tree is unfolded depth-first, and we
implement an early stopping criterion: 
If a surrogate triangle and a triangle set from another particle do not collide,
\deleted[id=Add]{then} the children of the surrogate triangle in our triangle
tree cannot collide either. 
The depth-first traversal along this branch of the tree thus can terminate
early.
The surrogate model \deleted[id=Add]{used ``refines''
(}unfolds\deleted[id=Add]{)} adaptively.

\begin{algorithm}[htb]
  \caption{
    Multiresolution contact detection within explicit time stepping. 
    It compares two particles $p_i$ and $p_j$ given by their surrogate trees
    $\mathcal{T}(p_i)$ and $\mathcal{T}(p_j)$ with each other.
    \label{algorithm:multiresolution-contact-detection:explicit-Euler}
    }
 {\footnotesize
 \begin{algorithmic}[1]
  \State $\mathbb{A}_i \gets root(\mathcal{T}(p_i))$,
         $\mathbb{A}_j \gets root(\mathcal{T}(p_j))$
   \Comment Set of active triangles to check
  \While{$ \mathbb{A}_i \not= \emptyset \vee \mathbb{A}_j \not= \emptyset  $}
   \State $\mathbb{A}_{i,\text{new}} \gets \emptyset$,
          $\mathbb{A}_{j,\text{new}} \gets \emptyset$
   \For{$t_i \in \mathbb{A}_i, t_j \in \mathbb{A}_j$}
    \State $c \gets \textsc{contact\_iterative}(t_i,t_j)$
     \Comment{Use context-specific $\epsilon $ depending on $t_i,t_j$}
    \If{$c = \odot \wedge t_i \in \mathbb{T}^{\epsilon}_0(p_i) \wedge t_j \in
     \mathbb{T}^{\epsilon}_0(p_j)$}
      \Comment{Not converged on non-surrogate triangles}
     \State $c \gets \textsc{contact\_comparison}(t_i,t_j)$
      \Comment{Use comparison-based algorithm this time}
    \EndIf
    \If{$c\not = \bot$}
      \If{$t_i \in \mathbb{T}^{\epsilon}_0(p_i) \wedge t_j \in
       \mathbb{T}^{\epsilon}_0(p_j)$}
       \Comment{No surrogate triangles, }
       \State $\mathbb{C}(p_i) \gets \mathbb{C}(p_i) \cup \{c\}$,
              $\mathbb{C}(p_j) \gets \mathbb{C}(p_j) \cup \{c\}$ 
       \Comment{i.e.~proper contact point}
      \Else
        \Comment Unfold
        \If{$t_i \in \mathbb{T}^{\epsilon}_0(p_i)$}
         \State $\mathbb{A}_{i,\text{new}} \gets \mathbb{A}_{i,\text{new}} \cup
          \{ t_i \}$
        \Else
         \State $\mathbb{A}_{i,\text{new}} \gets \mathbb{A}_{i,\text{new}} \cup
          \{\hat t: \hat t \sqsubseteq _{\text{child}} t_i\}$
        \EndIf 
        \If{$t_j \in \mathbb{T}^{\epsilon}_0(p_j)$}
         \State $\mathbb{A}_{j,\text{new}} \gets \mathbb{A}_{j,\text{new}} \cup
          \{ t_j \}$
        \Else
         \State $\mathbb{A}_{j,\text{new}} \gets \mathbb{A}_{j,\text{new}} \cup
          \{\hat t: \hat t \sqsubseteq _{\text{child}} t_j\}$
        \EndIf 
      \EndIf 
    \EndIf 
   \EndFor
   \State $\mathbb{A}_{i} \gets \mathbb{A}_{i,\text{new}}$,
          $\mathbb{A}_{j} \gets \mathbb{A}_{j,\text{new}}$
  \EndWhile
 \end{algorithmic}
 }
\end{algorithm}

%
%
The concept defines a marker
  (Algorithm~\ref{algorithm:multiresolution-contact-detection:explicit-Euler}):
Let $\mathbb{A} $ identify a set of active
nodes from $\mathcal{T}$. 
The union of all sets \replaced[id=Add]{labelled}{identified} by $\mathbb{A} $
yields all triangles from a particle that participate in collision checks.
At the begin of a particle-to-particle comparison, only the particles' roots are
active.
From there, we work our way down into finer and finer geometric representations
as long as the surrogate models suggest that there might be some contacts, until
we eventually identify real contact points stemming from the finest mesh.

%
%

\begin{lemma}
 The hierarchical algorithm yields exactly the same outcome as our baseline code
 over sets $\mathbb{T}^{\epsilon}(p_i)$ and $\mathbb{T}^{\epsilon}(p_j)$.
 Algorithm~\ref{algorithm:multiresolution-contact-detection:explicit-Euler} is correct.
 \label{lemma:multiresolution:explicit}
\end{lemma}

\begin{proof}
The argument relies on three properties:
\begin{enumerate}
  \item If a contact point is identified for a surrogate triangle, it is not
  added to the set of contact points. Therefore, a given active set never
  identifies artificial/too many contact points.
  \item A contact point is added if it stems from the comparison of
  two triangles from the fine grid tessellations which are in the active sets.
  \item Let two triangles $t_i^{\epsilon }$ and
  $t_j^{\epsilon }$ yield a contact point.
  \replaced[id=Add]{
   As surrogates are conservative,
  }{
   As of Algorithm \ref{algorithm:construct_surrogate_tree} (or
   the more generic definition of conservatism for surrogates),
  }
  they belong to
  nodes (triangle sets) $\mathbb{T}(p_i)$ and
  $\mathbb{T}(p_j)$ with 
  $\mathbb{T}(p_i) \sqsubseteq _{\text{child}} \hat t_i$ and $\mathbb{T}(p_j)
  \sqsubseteq _{\text{child}} \hat t_j$
  in $\mathcal{T}(p_i)$ or $\mathcal{T}(p_j)$, respectively. These
  surrogates fulfil
  \[
   t_i^{\epsilon } \cap t_j^{\epsilon } \not= \emptyset \Rightarrow
   \hat t_i^{\epsilon } \cap \hat t_j^{\epsilon } \not= \emptyset. 
  \]
  and therefore are replaced in the active set
  by their children \replaced[id=Add]{in}{by the} Algorithm
  \ref{algorithm:algorithmic-framework:explicit-Euler} before the respective
  algorithm terminates.
\end{enumerate}
\noindent
The correctness of the algorithm follows from bottom-up induction
over the 
levels of $\mathcal{T}$:
The property holds directly for the finest surrogate levels $\mathbb{T}_1$ of
the tree.
Any violation thus has to arise from $\mathbb{T}_k, k \geq 2$ in $p_i$ or $p_j$.
We apply the arguments recursively.
\end{proof}

\noindent
We have two triangle-to-triangle comparison strategies on the table (hybrid and
comparison-based) which are robust, i.e.~always yield the correct solution.
If we employ the comparison-based approach only, the $c = \odot$ condition never
holds and the corresponding branch is never executed.
Otherwise, our algorithmic blueprint implements the hybrid's fall-back as it
automatically re-evaluates the contact search for $c = \odot$.
However, it is indeed sufficient to rerun this a posteriori contact search if
and only if both triangles stem from the finest triangle discretisation:

\begin{corollary}
 \label{corollary:multiresolution:iterative}
 On the surrogate levels within the tree, it is sufficient to use
 the (efficient) iterative collision detection algorithm (Algorithm
 \ref{algorithm:triangle-soup}, bottom), without falling back to the
 comparison-based variant.
\end{corollary}

\begin{proof}
 Let $\mathbb{T}^\epsilon(p_i) \cap \mathbb{T}^\epsilon(p_j) \not= \emptyset$, i.e.~two particles collide.
 We assume the lemma is wrong, i.e.~the tree unfolding terminates prematurely.
 This assumption formally means
  \[
   \exists t_i \in \mathcal{T}(p_i), t_j \in \mathcal{T}(p_j): r(p_i,p_j) =
   \bot,
  \]
  \noindent
  with 
  \[
   \exists t_{0,i} \in \mathbb{T}(p_i), t_{0,j} \in \mathbb{T}(p_j): \
   t_{0,i} \sqsubseteq _{\text{child}} \ldots \sqsubseteq _{\text{child}} t_i
   \wedge
   t_{0,j} \sqsubseteq _{\text{child}} \ldots \sqsubseteq _{\text{child}}
   t_j
   \wedge
   t_i^{\epsilon} \cap t_j^{\epsilon} \not= \emptyset.
  \]
 This assumption is a direct violation of the definition of a surrogate model
 which has to be conservative.
\end{proof}

\subsection{Implicit Euler with multiresolution acceleration}
\label{subsection:implicit:multiresolution-acceleration}

%
%
Picard iterations can exploit the multiscale hierarchy by looping
over the hierarchy levels top down:
Per iteration of Algorithm \ref{algorithm:algorithmic-framework:implicit-Euler},
we have to identify all contact points for the current
particle configuration.
This search for contact points is the same search as we use it in an explicit
Euler.
If we replace the contact detection within the inner loop with our multiscale
contact detection from Section~\ref{subsection:multiscale:explicit-Euler}, we
obtain an implicit Euler where the surrogate concept is used \emph{within the Picard loop} as
multiresolution acceleration.
The surrogate concept enters the algorithm's implementation as a black-box.

\begin{corollary}
  An implicit Euler using surrogates within the Picard loop body to speed up the
  search for contact points yields the same output as a flat
 implicit code with the same number of Picard iterations.
 \label{corollary:multiresolution:implicit:acceleration}
\end{corollary}

\begin{proof}
 This is a direct consequence of Lemma  \ref{lemma:multiresolution:explicit}
 and implies the algorithm's correctness.
\end{proof}

\noindent
Though we end up with exactly the same number of Picard iterations,
the individual \replaced[id=Add]{iterations}{iterates} are accelerated
\deleted[id=Add]{internally} by the multiresolution technique:
\replaced[id=R1]{
 For a localised contact between two particles, the surrogate tree is
 unfolded along a single or few branches of the tree.
 If the nodes within the tree hold roughly the same number of triangles,
 the number of triangles to be compared grows linearly with the number of Picard
 steps.
 We benefit both from a zapping through the resolution levels and the 
 localisation of contacts, i.e.~the fact that two particles usually collide
 only in a small area compared to the overall geometric object.
}{
Per Picard step, we expect the surrogate trees' height times
$N_{\text{surrogate}}$ to dominate the compute cost---instead of $N^2_{\text{surrogate}}$ a plain code's.
}

\subsection{Implicit multi-resolution Euler}

%
%
A more bespoke implicit multiresolution algorithm arises from
\deleted[id=Add]{two} ideas inspired by multilevel non-linear equation system
solvers.
\replaced[id=Add]{The}{On the one hand, the} multiscale Algorithm
\ref{subsection:implicit:multiresolution-acceleration} consists of two nested
while loops---the outer loop stems from the Picard iterations, the inner loop
realises the tree unfolding---which we can permute.
We obtain an algorithm that runs top-down via the active sets through the
surrogate hierarchies and unfolds the trees step by step.
Per unfolding step, it uses the Picard loop to converge on the selected
hierarchy level.
The rationale behind such a permutation is the observation that the efficiency
of a nonlinear equation system solver hinges on the availability of a good
initial guess.
Surrogate resolution levels might be well-suited to deliver a good initial
guess of what $\mathbb{T}$ looks like in the next time step.
This train of thought is similar to the extension of multigrid into full
multigrid.
\replaced[id=Add]{The}{
On the other hand, the} same multigrid analogy suggests that we do not have to
converge on a surrogate level, as the level supplements only a guess anyway.
In the extreme case, it is thus sufficient to run one Picard iteration per
unfolding step only.

\begin{algorithm}[htb]
  \caption{
   Implicit time stepping algorithm where the
   Picard and multiresolution loop are intermingled.
   \label{algorithm:multiresolution-contact-detection:implicit-Euler}
  }
 {\footnotesize
 \begin{algorithmic}[1]
  \State $\forall p_j \not= p_i \in \mathbb{P}: \ $
         $\mathbb{A}(p_i,p_j) \gets root(\mathcal{T}(p_i))$
   \Comment Active sets are now parameterised over interactions 
  \While{$\mathbb{T}^{\text{guess}}(p_i), v^{\text{guess}}(p_i),
  \omega^{\text{guess}}(p_i)$ or any $\mathbb{A}$ change significantly for
  any $p_i$}
   \State $\forall p_j \not= p_i \in \mathbb{P}: \ $
          $\mathbb{A}_{\text{new}}(p_i,p_j) \gets \emptyset$,
          $\mathbb{A}_{\text{new}}(p_j,p_i) \gets \emptyset$,
          $\mathbb{C}(p_i,p_j) = \emptyset$,
          $\mathbb{C}(p_j,p_i) = \emptyset$
   \For{$p_j \not= p_i \in \mathbb{P}$}
    \For{$t_i \in \mathbb{A}(p_i,p_j), t_j \in \mathbb{A}(p_j,p_i)$}
    \State $c \gets \textsc{contactIterative}(t_i,t_j)$
     \Comment{Use context-specific $\epsilon $ depending on $t_i,t_j$}
    \If{$c = \odot \wedge t_i \in \mathbb{T}^{\epsilon}_0(p_i) \wedge t_j \in
     \mathbb{T}^{\epsilon}_h(p_j)$}
      \Comment{Not converged on non-surrogate triangles}
     \State $c \gets \textsc{contactComparisonBased}(t_i,t_j)$
      \Comment{Use comparison-based algorithm this time}
    \EndIf
    \If{$c\not = \bot \wedge c \not= \odot$}
       \Comment{Implicit guess}
       \State $\mathbb{C}(p_i) \gets \mathbb{C}(p_i) \cup \{c\}$, 
              $\mathbb{C}(p_j) \gets \mathbb{C}(p_j) \cup \{c\}$ 
    \EndIf
    \If{$c = \bot$}
        \Comment Add only parents
      \State 
        $\mathbb{A}_{\text{new}}(p_i,p_j) \gets \mathbb{A}_{\text{new}}(p_i,p_j)
        \cup \{\hat t: t_i \sqsubseteq _{\text{child}} \hat t\}$
      \State
        $\mathbb{A}_{\text{new}}(p_j,p_i) \gets \mathbb{A}_{\text{new}}(p_j,p_i)
        \cup \{\hat t: t_j \sqsubseteq _{\text{child}} \hat t\}$
    \Else
        \Comment Widen active sets
      \State \ldots
        \Comment Compare to Algorithm
        \ref{algorithm:multiresolution-contact-detection:explicit-Euler}
    \EndIf 
   \EndFor
   \EndFor
   \State $\forall p_j \not= p_i \in \mathbb{P}: \ $
          $\mathbb{A}(p_i,p_j) \gets \mathbb{A}_{\text{new}}(p_i,p_j)$,
          $\mathbb{A}(p_j,p_i) \gets \mathbb{A}_{\text{new}}(p_j,p_i)$
   \For{$p_i \in \mathbb{P}$}
     \State $(dv,d\omega) \gets$ \Call{calcForces}{$\mathbb{C}(p_i)$}
     \Comment Remove redundant contact points first
     \State $(v^{\text{guess}},\omega^{\text{guess}})(p_i) \gets (v,\omega)(p_i) +
      \Delta t \cdot (dv,d\omega)$ 
     \Comment Additional damping might be required
     \State $\mathbb{T}^{\text{guess}}(p_i) \gets $
      \Call{update}{$\mathbb{T}(p_i),v^{\text{guess}}(p_i),\omega^{\text{guess}}(p_i),\Delta
      t$}
   \EndFor
  \For{$p_j \not= p_i \in \mathbb{P}$}
   \Comment Clean-up, i.e. add siblings
   \State
   \Comment Obsolete if surrogate nodes host only one triangle 
   \State $\forall t, \hat t \in \mathbb{A}(p_i,p_j) \ \text{with} \ 
     t \sqsubseteq _{\text{child}} \hat t: $
     $\mathbb{A}(p_i,p_j) \gets \mathbb{A}(p_i,p_j)
       \cup
       \{ t' \in \mathcal{T}(p_i): t' \sqsubseteq _{\text{child}} \hat t \} $
   \State $\forall t, \hat t \in \mathbb{A}(p_j,p_i) \ \text{with} \ 
     t \sqsubseteq _{\text{child}} \hat t: $
     $\mathbb{A}(p_j,p_i) \gets \mathbb{A}(p_j,p_i)
       \cup
       \{ t' \in \mathcal{T}(p_j): t' \sqsubseteq _{\text{child}} \hat t \} $
   \State
   \Comment and remove ``redundant'' parents
   \State $\forall t \in \mathbb{A}(p_i,p_j): \ 
       \mathbb{A}(p_i,p_j) \gets \mathbb{A}(p_i,p_j)
       \setminus
       \{ \hat t \in \mathcal{T}(p_i): t \sqsubseteq _{\text{child}} \hat t \} $
   \State $\forall t \in \mathbb{A}(p_j,p_i): \ 
       \mathbb{A}(p_j,p_i) \gets \mathbb{A}(p_j,p_i)
       \setminus
       \{ \hat t \in \mathcal{T}(p_j): t \sqsubseteq _{\text{child}} \hat t \} $
  \EndFor
  \EndWhile
  \State $\forall p_i \in \mathbb{P}: \ \mathbb{T}(p_i) \gets
   \mathbb{T}^{\text{guess}}(p_i),\ v(p_i) \gets v^{\text{guess}}(p_i),\ 
     \omega(p_i) \gets \omega^{\text{guess}}(p_i) 
    $
  \State $t \gets t + \Delta t$
 \end{algorithmic}
 }
\end{algorithm}

%
%
Our advanced variant of the implicit Euler \deleted[id=Add]{thus}
is an \emph{outer-loop multi-resolution Picard scheme}.
Let the Picard loop start from the coarsest
surrogate representation per particle (Algorithm
\ref{algorithm:multiresolution-contact-detection:implicit-Euler}).
These representations form our initial active sets.
\deleted[id=R1]{
Different to the explicit scheme, we maintain an
active set $\mathbb{A}(p_i,p_j)$ per particle-particle
combination $p_i,p_j$:
A particle $p_i$ can exhibit a very coarse surrogate representation against one
particle, while \replaced[id=Add]{using}{use} a very detailed mesh when we
compare it to another one.
}
After the Picard step, any surrogate triangle for which the hybrid
algorithm has not terminated or for which we identified a contact point is
replaced by its next finer representation\replaced[id=Add]{.}{in the respective active set combination, i.e.~for the particular
comparison counterpart}.
In the tradition of value-range analysis, we widen the active set
\cite{Apinis:2016:WideningNarrowing}.
The Picard loop terminates if the plain algorithm's termination criteria hold,
i.e.~the outcome of two subsequent iterations does not change dramatically
anymore, and no surrogate tree node has unfolded anymore throughout the
previous iterate.

The algorithm is completed by a clean up which
removes ``\replaced[id=Add]{obsolete}{redundant}'' triangles from the active
set\replaced[id=Add]{:}{and ensures that the set is consistent with the tree:
It runs through the active set of a particle-particle
combination once again.
If any of a surrogate triangle's children is part of the active set, the
surrogate is removed from the set and the routine ensures that all of its
children are in the set.}
If all children of a surrogate triangle do certainly not contribute a contact
point anymore, they are \deleted[id=Add]{thus automatically} replaced with their
parent surrogate triangle.
We narrow the active set.
\deleted[id=R1]{Our algorithm discussion closes with the observation that the
number of particle-particle combination is potentially huge yet small in practice, as
particles are rigid and thus cannot cluster arbitrarily dense.}

\begin{implementation}
 Different to the explicit scheme, we maintain an
 active set $\mathbb{A}(p_i,p_j)$ per particle-particle
 combination $p_i,p_j$:
 A particle $p_i$ can exhibit a very coarse surrogate representation against one
 particle, while \replaced[id=Add]{using}{use} a very detailed mesh when we
 compare it to another one.
 While the number of particle-particle combinations is potentially huge, it is 
 small in practice, as particles are rigid and thus cannot cluster arbitrarily dense.
\end{implementation}

\noindent
Our genuine multiscale formulation stresses the convergence
\replaced[id=Add]{assumptions:}{requirements further}:
While Assumption \ref{assumption:contraction} guarantees the convergence of the
Picard iterations on the finest level, 
our multi-resolution approach may push the solution into the wrong direction via
the surrogate levels and thus make the initial guess on the next finer level
leave the single level's convergence domain.

\begin{assumption}
 \label{assumption:multiscale-contraction}
 We assume that a Picard iteration on any level of the surrogate trees yields a
 new solution on the same or a finer resolution which
 \replaced[id=Add]{preserves the Picard iteration's contraction
 property.}{remains inside the respectve Picard
 iteration's region of convergence.}
\end{assumption}

\deleted[id=Add]{
In practice, Assumption \ref{assumption:multiscale-contraction} might require
a damping of the Picard iterations
with a relaxation parameter $\theta _{\text{Picard}} \leq 1$
such that an iteration update does not overshoot.
That is, the damping becomes stronger with coarser surrogate levels.
}

\begin{lemma}
 \label{lemma:implicit-multiresolution:correctness}
  If Assumption \ref{assumption:multiscale-contraction} holds and if Algorithm
  \ref{algorithm:multiresolution-contact-detection:implicit-Euler}
  terminates, it delivers the correct solution.
\end{lemma}

\begin{proof}
  We have to study two cases over  
  the active sets 
  $ \mathbb{A}(p_i,p_j)$ and $\mathbb{A}(p_j,p_i)$.
  First, assume that 
  $ (t_{\mathbb{A}(p_i,p_j)}, t_{\mathbb{A}(p_j,p_i)})
  \in \mathbb{A}(p_i,p_j) \times \mathbb{A}(p_j,p_i)$
  yields an invalid contact point, i.e.~a
  contact point
  that does not exist in $\mathbb{T}_0(p_i)$ compared to $\mathbb{T}_0(p_j)$.
  One of the triangles has to be a surrogate triangle.
  They are
  replaced by their children and the algorithm has not
  terminated.
  Instead, we approach the solution further.

  In the other case, assume that the  
  algorithm has terminated yet misses a
  triangle pair 
  $ (t(p_i), t(p_j)) \not \in \mathbb{T}(p_i) \times \mathbb{T}(p_j)$
  which contributes a contact point in the plain model.
  Due to Definition 
  \ref{definition:multiresolution:conservative} over conservative surrogates,  
  
\begin{equation}
 \forall t(p_i) \in \mathbb{T}_0(p_j), \exists \hat t_{\mathbb{A}(p_i,p_j)} \in
 \mathbb{A}(p_i):
 \quad t(p_i) \sqsubseteq _{\text{child}} \ldots \sqsubseteq
 _{\text{child}} \hat t_{\mathbb{A}(p_i,p_j)}
 \label{equation:05c_implicit:inclusion}
\end{equation}
  such that $\hat t_{\mathbb{A}(p_i,p_j)}$ yields a contact point.
  This point is \deleted[id=Add]{``invalid'', i.e.~not found in the baseline and
  thus covered by the first case.
  It} eventually \deleted[id=Add]{is} removed as the corresponding
  $\hat t_{\mathbb{A}(p_i,p_j)}$ is replaced by its
  children.
  \replaced[id=Add]{The algorithm has not terminated yet.}{The analogous
  argument can be made over $t_\mathbb{U}$.}
\deleted[id=Add]{
 Both case distinctions argue over the widening of the search space 
 $ \mathbb{A}^{(n)}(p_j) \times \mathbb{U}^{(n)}(p_j) $.
 The modification of this space by the algorithm implies that the comparison
 sequence 
 $ \mathbb{A}_0^{(n)}(p_j) \times \mathbb{U}_0^{(n)}(p_j) $ has to be changed
 after the respective modification, too.
 As we assume that we remain within the region of convergence of the Picard
 iteration, this harms the convergence speed but does not imply that we diverge.
 Little additional work is eventually required to handle the narrowing case:}
 \end{proof}

\begin{lemma}
 \deleted[id=Add]{
  If Algorithm
  \ref{algorithm:multiresolution-contact-detection:implicit-Euler}
  remains within the region of convergence of the Picard
  iteration, it terminates.
 }
\end{lemma}

\begin{corollary}
 \label{corollary:implicit-multiresolution:correctness}
  The removal of triangles from the active set can cause Algorithm 
  \ref{algorithm:multiresolution-contact-detection:implicit-Euler} to violate 
  Assumption \ref{assumption:multiscale-contraction}.
\end{corollary}

\begin{proof}
  If no triangles are ever removed from the active set, the proof of Lemma
  \ref{lemma:implicit-multiresolution:correctness} trivially demonstrates that
  the algorithm terminates always, as the surrogate tree is of finite depth and
  width.
  Even if we overshoot with the Picard iterations, i.e.~if we violate the
  contraction property, we will, in the worst case, get
  $\mathbb{A}(p_i,p_j) = \mathbb{T}(p_i)$.
  From hereon, the algorithm converges.

  If we however remove triangles, it is easy to see that we cannot guarantee
  that we do not introduce cycles or even amplify oscillations.
  The contraction property is violated.
\end{proof}


\begin{proof}
\deleted[id=Add]{
  Our discussion of Lemma \ref{lemma:implicit-multiresolution:correctness}
  assumes a monotonous growth of 
  $ \mathbb{A}^{(n)}(p_j) \times \mathbb{U}^{(n)}(p_j) $
  and exploits the fact that this search space is finite and bounded.
  Let $ (t_\mathbb{A},t_\mathbb{U}) \in \mathbb{A}^{(n)}_0(p_j) \times
  \mathbb{U}_0^{(n)}(p_j) $ not contribute a contact point.
  Neither does its parent $(\hat t_\mathbb{A},t_\mathbb{U})$ or $(
  t_\mathbb{A},\hat t_\mathbb{U})$, respectively, or any other child of this
  parent contribute a contact point.
  $ (t_\mathbb{A},t_\mathbb{U}) $ consequently is replaced by a combination
  involving its surrogates in $\mathbb{A}^{(n+1)}_0(p_j) \times
  \mathbb{U}_0^{(n+1)}(p_j)$.
  The argument applies recursively.
  Let there be a $m>n+1$ for which  $ (t_\mathbb{A},t_\mathbb{U}) $ has to be
  taken into account.
  We know that it will eventually be re-added.
  While the cardinalities $\| \mathbb{A}^{(n)}(p_j)\| $ and
  $\| \mathbb{U}^{(n)}(p_j)\| $ are not monotonously growing, they are
  non-strictly growing between iteration $n$ and $m$.
  Furthermore, we know that the error behind the Picard loop has (strictly)
  diminished between $n$ and $m$ due to contraction property.
  We do not encounter cycles.
  }
\end{proof}

\begin{implementation}
 In practice, Corollary \ref{corollary:implicit-multiresolution:correctness}
 implies that we, on the one hand, have to damp the Picard iterations.
 We artificially reduce the force contributions from coarse surrogate levels to
 avoid oscillations.
 On the other hand, we work with a memory set $\mathbb{U}(p_i,p_j)$ in which we
 hold references to triangles which have been removed from the active set.
 Once they are re-added, we veto any subsequent removal from hereon and, hence, 
 the activation of such triangles' surrogates. 
\end{implementation}

\subsection{Implementation}
\deleted[id=Add]{The multiresolution representation of an object can be computed
at simulation startup as a preprocessing step.
As we keep the multiresolution hierarchy when
particles move and rotate---we simply have to ensure that all triangles including all
surrogates are properly rotated and translated---the flattening
of sets of triangles from $\mathcal{T}(p)$ into a sequence of coordinates has to be done once per
time step, as the triangle coordinates change in each step.
It is reasonable to realise this via lazy flattening, i.e.~a given
set of triangles is mapped onto its flat representation---including the
replication of coordinates---upon first request and then cached for the
remainder of the time step.}
There are two reasons why our multi-resolution algorithms are expected to yield
better performance than \replaced[id=Add]{a straightforward textbook
implementation}{their baseline without a hierarchy}:
First and foremost, we expect the number of
triangle-to-triangle comparisons to go down \replaced[id=Add]{compared to a
flat, single-level approach.}{ despite the fact that we augmented the triangle
set with surrogates.}
The multiscale algorithm iteratively narrows down the region of a particle where
contacts may arise from\deleted[id=Add]{, and thus studies only the area of a
particle which potentially is in contact with a neighbour}.
\replaced[id=Add]{
 These savings on the finer geometric resolutions compensate for additional
 checks with surrogate triangles.
 However, any cost amortisation has to be studied carefully---in particular for
 the implicit, non-linear case where trees unfold and collapse again---and it
 hinges upon an efficient realisation: 
 In this context, we expect the streaming, comparison-free variants of our
 algorithm to benefit from vector architectures.
}{
At the same time, we can pick $N_{\text{surrogate}}$ such that one leaf set
cardinality of the surrogate tree fits exactly to the vector unit length and,
hence, cache line architecture.
}

The multiresolution representation of an object can be computed at simulation
startup as a preprocessing step.
\replaced[id=Add]{Though}{While} we keep the multiresolution hierarchy when
particles move and rotate, \deleted[id=Add]{---we simply have to ensure that all
triangles including all surrogates are properly rotated and translated---}the
flattening of \added[id=Add]{the active} sets of triangles from $\mathcal{T}(p)$
into a sequence of coordinates \replaced[id=Add]{is done on-the-fly and the
flattened data is not held persistently.}{has to be done, as the triangle
coordinates change in each step.
It is reasonable to realise this via lazy flattening, i.e.~a given
set of triangles is mapped onto its flat representation---including the
replication of coordinates---upon first request and then cached for the
remainder of the time step.}

 On the one hand, this ensures that the memory overhead remains under control. 
 On the other hand, it pays tribute to the fact that the active set changes
 permanently.
 To remain fast despite permanently changing active sets, we pick
 $N_{\text{surrogate}}$ such that the finest nodes within $|\mathbb{T}|$ hold
 triangle sequences for which streaming instructions such as AVX already pay
 off.
 The tree clusters $\mathbb{T}$ into segments that fit to the architecture, and 
 Algorithm \ref{algorithm:triangle-soup} hence does not process all triangles
 from $\mathbb{T}$ in one batch.
 Instead it runs over subchunks of batches.

 We apply this argument recursively and make each non-leaf node within
 $\mathcal{T}$ hold a set of triangles, too.
 We make the nodes in the surrogate tree host many triangles and the tree
 overall shallow, such that the per-node data cardinality again ensures that we
 benefit from vector units.
 This ``tweak $N_{\text{surrogate}}$'' idea however does not fit perfectly to
 multiscale algorithms
where the coarser tree levels\replaced[id=Add]{ typically}{, by definition,} do
not occupy a complete vector length.
 It would be a coincidence if $|\mathbb{T}|$ and $N_{\text{surrogate}}$ 
\added[id=Add]{
 yielded only nodes that fill a vector unit completely on each and every
 surrogate level.
} 
Therefore, we do not run triangle-to-triangle comparisions within the
 tree directly.
Instead, we make the tree/triangle traversal collect all comparisons to
be made with a buffer.
Once we have identified all triangle collisions to be computed, we stream the
whole buffer through the vector units.
We merge the \deleted[id=Add]{flattened} triangle representations on-the-fly.

\begin{implementation}
  If our surrogate tree hosts more than one triangle per non-leaf node, the
  algorithm has to be \deleted[id=R3]{is} completed by a further clean-up step which ensures that
  the active set remains consistent with the tree:
  It runs through the active set of a particle-particle
  combination once again.
  If any of a surrogate triangle's children is part of the active set, the
  surrogate is removed from the set, but all of its
  children become active.
  Without such an additional sweep over the active set, all triangles of a node
  could be active plus the children of one triangle which implies that we test
  against the children triangle set plus their surrogate triangle.
  Restricting the node triangle cardinality for surrogate levels to one renders
  the additional clean-up unneccessary. 
\end{implementation}

  \section{\replaced[id=R1]{Runtime results}{Results}}
\label{section:results}
 Our algorithms yield correct results (Lemma \ref{lemma:multiresolution:explicit}, Corollary 
 \ref{corollary:multiresolution:implicit:acceleration} and Lemma
 \ref{lemma:implicit-multiresolution:correctness}), but they do not provide an
 efficiency guarantee.
We hence collect runtime results\added[id=Add]{, i.e.~gather empirical
evidence. Real-world experiments benchmarked against measurements remain
out-of-scope for the present paper. We furthermore continue to focus on the
actual collision detection and neglect the impact of different time step
sizes---in particular comparisions between implicit and explicit schemes
facilitating different stable time step choices---the cost of a coarse-grain
neighbour search via a grid, e.g., and notably the construction cost for the
surrogates which are done offline prior to the simulation run}.
All experiments are \replaced[id=Add]{run}{ran} on Intel Xeon E5-2650V4
(Broadwell) chips in a two socket configuration with $2 \times 12$ cores.
They run at 2.4~GHz, though
TurboBoost can increase this up to 2.9~GHz.
However, a core
executing AVX(2) instructions will fall back to a reduced frequency (minimal
1.8~GHz) to stay within the TDP limits
\cite{Charrier:2019:Energy}.

Our node has access to 64~GB TruDDR4 memory, which is connected via 
a hierarchy of three inclusive caches.
They host $12\times (32+32)$ KiB, $12 \times 256$ KiB or $12 \times 2.5$ MiB,
respectively.
We obtain around 109 GB/s in the Stream TRIAD
\cite{McCalpin:1995:Stream} benchmark on the node which translates into 
4,556 MB/s per core. 
The node has a theoretical single precision
peak performance between $2.4$ (non-AVX mode and baseline speed)
and $46.4$ Gflop/s per core (\texttt{AVX 2.0 FMA3} with full turbo
boost).
All of our calculations are ran in single precision.
They are translated with the Intel 19 update 2 compiler and use
the flags \texttt{-std=c++17 -O3 -qopenmp -march=native -fp-model fast=2},
i.e.~we tailor them to the particular instruction set.

All presented performance counter data are read out through LIKWID \cite{Treibig:2010:Likwid}.
DEM codes are relatively straightforward to parallelise as their
particle-particle interaction is strongly localised: 
We can combine grid-based parallelism (neighbour cells) with an additional
parallelisation over the particle pairs \cite{Krestenitis:17:FastDEM}.
The load balancing of these concurrency dimensions however remains challenging.
As our ideas reduce the comparison cost algorithmically yet do not alter
the concurrency character, we stick (logically) to single core experiments to
avoid biased measurements due to parallelisation or load balancing overheads.
Yet, we artificially scale up the setup by replicating the computations per
node over multiple OpenMP threads whenever we present real runtime data or
machine characteristics, and then break down the data again into cost per
replica per core.
This avoids that simple problems fit into a particular cache or that
memory-bound applications have exclusive access to two memory
controllers.

\subsection{Experimental setup}
\replaced[id=R2]{
 We work with two simple benchmark setups, before we validate our
 results for large numbers of particles.
}{ 
 We work with two different
 experimental setups. 
} In the \emph{particle-particle}
setup, we study two spherical objects which are set on direct collision trajectory.
They bump into each other, and then separate again.
The setup yields three computational phases:
While the particles approach, there is no collision and no forces act on the
particles as we neglect gravity.
When they are close enough, the particles exchange forces and the system becomes
very stiff suddenly, before the objects repulse each other again and separate.
We focus exlusively on the middle phase.
Throughout this approach-and-contact situation, the algorithmic complexity of
the contact detection is in 
$\mathcal{O}(|\mathbb{T}|^2)$, as we assume that both particles have
the same triangle count.

In the \emph{particle-on-plane
scenario}, we drop a spherical object onto a tilted
\replaced[id=Add]{plane}{plate}.
The particle hits the \replaced[id=Add]{plane}{plate}, bumps back in a slightly tilted angle, i.e.~with a
rotation, and thus hops down the \replaced[id=Add]{plane}{plate}.
This problem yields free-fall phases which take turns with stiff in-contact
situations.
Furthermore, the area of the free particles which is subject to potential
contacts changes all the time as the particle starts to rotate,
and the contacts result from a complex geometry consisting of many triangles compared to a simplistic
geometry with very few triangles.
The underlying computational complexity is  in $\mathcal{O}(|\mathbb{T}|)$.

 In the \emph{grid scenario}, \added[id=R2]{we finally arrange 24,576 spheres in
 a Cartesian grid. Each particle slightly overlaps the epsilon region of it's
 neighbours. As there is no ground plane or gravity, the particles ``float''
 in space.
 Due to the regular particle layout, the interaction pattern 
 yields a Cartesian topology, i.e.~each particle collides with four other
 particles initially.
}

\replaced[id=Add]{Our codes}{For the particles, we} work exclusively with
\emph{sphere-like particle shapes}, which result from a randomised parameterisation:
We decompose the sphere with radius 1 into
$|\mathbb{T}|$ triangles. 
If not stated otherwise,
$|\mathbb{T}|=1,280$.
\added[id=R2]{In the first two scenarios} the vertices on the sphere which span the triangles are 
subject to a Perlin noise function, which offsets the vertex along the normal direction of the surface.
\replaced[id=Add]{$\eta _r=1$}{$\eta _r=0$} adds no noise and thus yields a perfect, triangulated sphere \added[id=Add]{with radius 1} where
all vertices are exactly $1$ unit away from the sphere's origin.
Otherwise, the per-vertex radius is from
\replaced[id=Add]{$[1,\eta _r]$}{$[1,1+ \eta _r]$}.
As we use a hierarchical noise model, a high $\eta _r $ yields a degenerated
shape which retains a relatively smooth surface.

 For the implicit schemes, we consider the result converged when the
 update to the force and torque applied to every particle underruns a relative
 threshold of 1\%.
 With this accuracy, single precision is sufficient.
 The Picard iterations are subject to damping and acceleration:
 Any update $(dv,d\omega)$ relative to the start configuration of a time step
 results from the weighted average between
 the currently computed forces and the forces of the previous step.
 If forces ``pull'' into one direction over multiple iterations, the 
 updates behind trials become successively bigger.
 If the forces oscillate, these oscillations are diminishing.
 Empirical data suggest that this choice helps us to meet Assumptions
 \ref{assumption:contraction} and \ref{assumption:multiscale-contraction}.
 $\epsilon = 10^{-2}$ is uniformly used on the finest mesh level.
 This is a relative quantity, i.e.~chosen relative to the particle diameter.
 For the plane, we uniformly use $\epsilon = 10^{-2}$.

\subsection{Surrogate properties}
%
%
We first assess our surrogate geometry's properties. 
Our coarsest surrogate model consists of a single triangle.
We compare this triangle's longest edge (diameter) $d_{k_{\text{max}}}$ plus
its corresponding $\epsilon _{k_{\text{max}}}$ value to the radius
$r_{\text{sphere}}=\frac{\eta _r}{2}$ of the bounding sphere of the fine grid object
(Table \ref{table:surrogate-properties:bounding-sphere}).
For the surrogate hierarchy, we use $N_{\text{surrogate}}=8$ as coarsening
factor; a choice we employ throughout the experiments.
 This implies that an object with $|\mathbb{T}|=64$ triangles spans three
 surrogate organised as a tree:
 They host $|\mathbb{T}|=64$ (original model),
 $|\mathbb{T_1}|=64/N_{\text{surrogate}}=8$ and
 $|\mathbb{T_2}|=1$ triangles.
In this first test, we \replaced[id=Add]{approximate low frequency noise by scaling along one axis}{keep the lowest frequency of the Perlin noise only},
i.e.~we \replaced[id=R1]{elongate}{stretch} the sphere along one direction yet
\replaced[id=Add]{do not introduce}{eliminate any further} bumps or extrusions. 
With growing $\eta _r$, we obtain increasingly non-spherical objects resembling
an ellipsoid.
The rationale behind this simplified noise is that we eliminate
non-deterministic effects and study the dominant sphere distortion effects.

\ifthenelse{ \boolean{presentAllData} }{
\begin{table}[htb]
 \caption{
  Different triangle counts
  \replaced[id=Add]{$|\mathbb{T}|$}{$N_{\text{triangles}}$} per spherish object
  plus different magnitudes of imposed Perlin noise $\eta _r$.
  Per setup, we study the top level surrogate which contains one
  triangle and compare the
  maximum triangle diameter plus its halo size against the
  bounding sphere diameter of the underlying geometry.
  \label{table:surrogate-properties:bounding-sphere}
 }
 \begin{center} 
 {\footnotesize
 \begin{tabular}{|r|rrr|rrr|rrr|}
  \hline
  & \multicolumn{3}{|c|}{$N_{\text{triangles}}=20$}
  & \multicolumn{3}{|c|}{$N_{\text{triangles}}=320$}
  & \multicolumn{3}{|c|}{$N_{\text{triangles}}=1,280$}
  \\
  $\eta _r $
  & $d_{k_{\text{max}}}$
  & $\epsilon _{k_{\tet{max}}}$
  & $d_{\text{sphere}}$
  & $d_{k_{\text{max}}}$
  & $\epsilon _{k_{\text{max}}}$
  & $d_{\text{sphere}}$
  & $d_{k_{\text{max}}}$
  & $\epsilon _{k_{\text{max}}}$
  & $d_{\text{sphere}}$
  \\
  \hline
  0.0
  & 0.29 & 0.72 & 1.00 & 0.38 & 0.95 & 1.00 & 0.95 & 0.84 & 1.00
  \\
  0.1
  & 0.29 & 0.78 & 1.20 & 0.64 & 0.86 & 1.20 & 0.30 & 1.06 & 1.20
  \\
  0.2
  & 0.28 & 0.77 & 1.40 & 0.48 & 0.93 & 1.40 & 0.40 & 0.99 & 1.40
  \\
  0.4
  & 0.44 & 0.89 & 1.80 & 0.54 & 1.28 & 1.40 & 0.54 & 1.04 & 1.80
  \\
  0.8
  & 0.77 & 0.99 & 2.60 & 1.16 & 1.21 & 2.60 & 1.60 & 1.37 & 2.60
  \\
  \hline
 \end{tabular}
 }
 \end{center}
\end{table}

\begin{table}[htb]
 \caption{
  Different triangle counts \replaced[id=Add]{$|\mathbb{T}|$}{$N_{\text{triangles}}$} per spherish object plus
  different magnitudes of imposed Perlin noise $\eta _r$.
  Per setup, we study the top level surrogate which contains one
  triangle and compare the
  maximum triangle diameter plus its halo size against the
  bounding sphere radius with the same origin as the underlying geometry expanded
  to encompass the maximum noise.
  \label{table:surrogate-properties:bounding-sphere}
 }
 \begin{center} 
 {\footnotesize
 \begin{tabular}{|r|rrr|rrr|rrr|}
  \hline
  &
  \multicolumn{3}{|c|}{\replaced[id=Add]{$|\mathbb{T}|=20$}{$N_{\text{triangles}}=20$}}
  &
  \multicolumn{3}{|c|}{\replaced[id=Add]{$|\mathbb{T}|=320$}{$N_{\text{triangles}}=320$}}
  &
  \multicolumn{3}{|c|}{\replaced[id=Add]{$|\mathbb{T}|=1,280$}{$N_{\text{triangles}}=1,280$}}
  \\
  $\eta _r $ 
  & $d_{k_{\text{max}}}$
  & $\epsilon _{k_{\text{max}}}$
  & $r_{\text{sphere}}$
  & $d_{k_{\text{max}}}$
  & $\epsilon _{k_{\text{max}}}$
  & $r_{\text{sphere}}$
  & $d_{k_{\text{max}}}$
  & $\epsilon _{k_{\text{max}}}$
  & $r_{\text{sphere}}$
  \\
  \hline
  0.0
  & 0.09 & 0.55 & 0.50 & 0.09 & 0.50 & 0.50 & 0.08 & 0.52 & 0.50
  \\
  0.1
  & 0.09 & 0.54 & 0.60 & 0.09 & 0.54 & 0.60 & 0.09 & 0.55 & 0.60
  \\
  0.2
  & 0.10 & 0.61 & 0.70 & 0.09 & 0.58 & 0.70 & 0.10 & 0.61 & 0.70
  \\
  0.4
  & 0.11 & 0.68 & 0.90 & 0.10 & 0.70 & 0.90 & 0.12 & 0.68 & 0.90
  \\
  0.8
  & 0.13 & 0.77 & 1.30 & 0.14 & 0.87 & 1.30 & 0.13 & 0.94 & 1.30
  \\
  \hline
 \end{tabular}
 }
 \end{center}
\end{table}

}
{

\begin{table}[htb]
 \caption{
  Different triangle counts
  \replaced[id=Add]{$|\mathbb{T}|$}{$N_{\text{triangles}}$} per spherish object
  scaled along one axis by a factor of $\mu$.
  Per setup, we study the top level surrogate which contains one
  triangle and compare the
  maximum triangle diameter plus its halo size against the
  bounding sphere radius.
  \added[id=Add]{Here we only report on the increase in $\epsilon$ for the coarsest surrogate ie.~at the finest level $\epsilon=0$.}
  \label{table:surrogate-properties:bounding-sphere}
 }
 \begin{center} 
 {\footnotesize
 \begin{tabular}{|r|rr|rr|rr|r|}
  \hline
  &
  \multicolumn{2}{c|}{\replaced[id=Add]{$|\mathbb{T}|=80$}{$N_{\text{triangles}}=80$}}
  &
  \multicolumn{2}{c|}{\replaced[id=Add]{$|\mathbb{T}|=320$}{$N_{\text{triangles}}=320$}}
  &
  \multicolumn{2}{c|}{\replaced[id=Add]{$|\mathbb{T}|=1,280$}{$N_{\text{triangles}}=1,280$}}
  &
  \\
  \replaced[id=Add]{$\eta _r$}{$\mu $} 
  & $d_{k_{\text{max}}}$
  & $\epsilon _{k_{\text{max}}}$
  & $d_{k_{\text{max}}}$
  & $\epsilon _{k_{\text{max}}}$
  & $d_{k_{\text{max}}}$
  & $\epsilon _{k_{\text{max}}}$
  & $r_{\text{sphere}}$
  \\
  \hline
  1.0
  & 0.09 & 0.49 & 0.09 & 0.50 & 0.08 & 0.52 & 0.50
  \\
  1.2
  & 0.10 & 0.55 & 0.10 & 0.56 & 0.10 & 0.56 & 0.60
  \\
  1.4
  & 0.34 & 0.54 & 0.12 & 0.65 & 0.11 & 0.66 & 0.70
  \\
  1.8
  & 0.14 & 0.84 & 0.89 & 0.53 & 1.35 & 0.51 & 0.90
  \\
  2.6
  & 1.54 & 0.58 & 2.24 & 0.51 & 2.36 & 0.52 & 1.30
  \\
  \hline
 \end{tabular}
 }
 \end{center}
\end{table}
}

\begin{figure}[htb]
 \begin{center}
  \includegraphics[width=.7\textwidth]{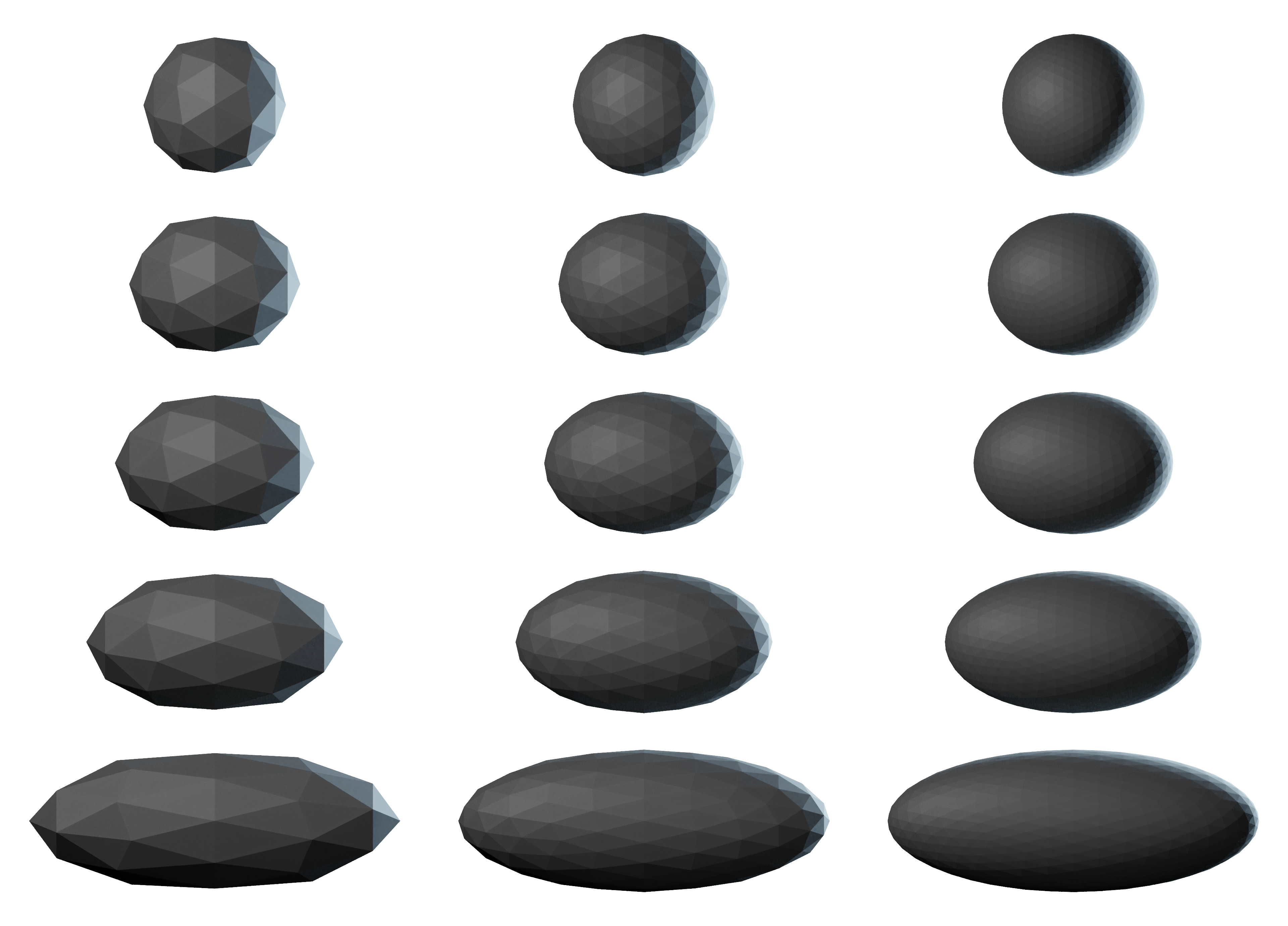}
 \end{center}
 \vspace{-0.3cm}
 \caption{
   A series of objects created from unit spheres with increasing level of detail
   ($80 \leq |\mathbb{T}| \leq 1,280$ from left to right)
   and a scale factor ($1.0 \leq \eta _r \leq 2.6$ from top to bottom).
  \label{figure:surrogate:sphereish}
 }
\end{figure}

The combination of $d_{k_{\text{max}}}$ and $\epsilon
_{k_{\text{max}}}$ characterises the shape of our coarsest surrogate model.
A large diameter relative to a small halo size describes a disc-like object.
A small diameter relative to a large halo size describes a sphere-like object. 
Different triangle counts for the fine grid model allow us to assess the
impact of the level of detail of the fine grid mesh onto the resulting
coarsest surrogate geometry.

%
%
Our surrogate model of choice on the coarsest level
(cmp.~the penalty term in (\ref{equation:surrogate-triangle}) of the Appendix)
almost degenerates to a point if the underlying triangulated geometry approximates a sphere.
\replaced[id=Add]{
 It approaches a bounding sphere.
}{
 It can not totally degenerate as we penalise triangle
 degeneration in (\ref{equation:surrogate-triangle}).
}
The triangle count approximating a spherical object does not have a significant qualitative or quantitative impact
on this characterisation of the coarsest surrogate triangle.
Once the triangulated mesh becomes less spherical, the surrogate triangle starts to
align with the maximum extension of the fine mesh.
It spreads out within the geometry along the geometry's longest diameter; 
an effect that is the more distinct the higher the fine geometry's triangle
count.
The halo layer $\epsilon _{k_{\text{max}}}$ around the surrogate triangle, 
which is analogous to a sphere's radius if the surrogate triangle approaches 
a point, remains in the order of $r=0.5$.
This is the radius of the original unit sphere (\replaced[id=Add]{$\eta _r = 1$}{$\eta _r = 0$}).

%
%
For a close-to-spherical geometry, our volumetric surrogate model
never exceeds 135\% of the bounding sphere volume ($|\mathbb{T}|=1,280$).
For the highly non-spherical cases ($\eta _r = 2.6$) our surrogate 
volume can be as little as $37\%$ ($|\mathbb{T}|=80$) 
of the simple bounding sphere volume.
This advantageous property results from the observation that a growth of
$d_{k_{\text{max}}}$ anticipates any extension of the geometry, 
while the $\epsilon _{k_{\text{max}}}$ ensures that the minimal geometry
diameter, which is at least as large as the original sphere, remains covered by the
surrogate triangle plus its halo environment.

\begin{observation}
 For highly non-spherical sets of triangles, our surrogate formalism yields advantageous
 representations.
 For spherical observations, it resembles the bounding sphere.
 This holds for all levels of the surrogate cascade.
\end{observation}

%
%

\noindent
In a surrogate tree, fine resolution tree nodes (surrogate triangles) are
characterised by the low triangle count measurements
in Table~\ref{table:surrogate-properties:bounding-sphere}
 where localised patches are highly non-spherical (large $\eta _r$). 
Surrogate triangles belonging to coarser levels inherit characteristics corresponding to
larger $|\mathbb{T}|$.
We conclude that our triangle-based multiresolution approach
is particularly advantageous as an early termination criterion (``there is certainly no
collision'') on the rather fine surrogate resolution levels within the
surrogate tree, or is overall tighter fitting than bounding sphere
formalisms for non-spherical geometries.

\subsection{Hybrid single level contact detection}
\label{subsection:hybrid-single-level}
%
%
Even though our multiscale approach intends to reduce the number of distance
calculations, a high throughput of the overall algorithm continues to 
hinge on the efficiency of the core distance calculation.
We hence continue with studies around the explicit Euler where we omit the
multiscale hierarchy.
We work with the finest particle mesh representation only.

%
%
The assessment of the core comparison efficiency relies on our
\replaced[id=Add]{sphere-on-plane}{sphere-on-plate} and the
\replaced[id=R1]{particle-particle}{two-particles} setup.
They represent two extreme cases of geometric comparisons:
With the plane, a complex particle with many triangles
hits very few triangles.
 As long as the triangles spanning the plane are large relative to the particle
 diameter, it is irrelevant how the plane is modelled, i.e.~if it consists of
 solely one or two triangles or an arrangement of multiple triangles.
 The multiscale algorithms assymptotically approach a $ | \mathbb{T} | : 1 $
 comparison with $\mathbb{T}$ given by the particle.
When we collide two particles---we use the same triangle count for both---we
 can, in the worst case, run into a $ | \mathbb{T} | : |
 \mathbb{T} | $ constellation.
 For both setups, we use strictly spherical particles ($\eta _r=1$).

%
%
\replaced[id=Add]{
 For the triangle-triangle comparisons,
}{
 On the algorithm side, 
}
two code variants are on the table:
We can run 
\replaced[id=Add]{
 the comparison-based (baseline) algorithm,
}{
 solely run a comparison-based algorith,
} or we can use the hybrid code
variant which runs four steps of the iterative scheme before it checks if the two last iterates of
the contact point differ by more than $C\epsilon _h$\replaced[id=Add]{.
 We use $C \approx 1$.
}{ with $C
\approx 1$ and $\epsilon _h = 0.01$; relative to a baseline particle diameter of 1.0 subject
to added noise.
}
If the \replaced[id=R3]{difference}{differerence} exceeds the threshold, our algorithm assumes that the code
has not converged,
and hence reruns the comparison-based code to obtain a valid
contact assessment.
The comparison-based code variant is 4-way vectorised and relies on Intel
intrinsics.
The hybrid variant is vectorised over batches of eight packed triangle
pairs using an OpenMP \texttt{simd} annotation.

\begin{table}[htb]
 \caption{
  \replaced[id=Add]{Particle-particle scenario.}{Particle collision scenario
  (top) and particle-on-plane (bottom)}.
  We compare a comparison-based realisation \added{(top)} against a hybrid realisation \added{(bottom)}.
  Per setup, we present the time-to-solution ([t]=s) per Euler step, i.e.~one
  run through all possible triangle combinations,
  and we augment these data with MFlop/s rates split up into scalar and
  vectorised contributions.
  Vector calculations are categorised as 128 bit packed (SSE) 
  or 256 bit packed (AVX) for four and eight simultaneous 
  32 bit floating point operations respectively.
  For the hybrid setup, we finally quantify how many triangle pairs had to be
  checked a posterio, i.e.~as fallback, by the comparison-based algorithm. 
  This runtime is included in the data.
  All measurements are given as the average per core.
  \label{table:hybrid-vs-comparison-particle-particle}
 }
 \begin{center}
 \footnotesize
 \begin{tabular}{|cr|rrrrr|}
\hline
 & & & & Packed & Packed &
 \\
 & $|\mathbb{T}|$
 & Runtime & Scalar & 128B & 256B & Fallback
 \\
\hline
\multirow{3}{*}{\STAB{\rotatebox[origin=c]{90}{Comp.}}} & 12 
 & $6.59 \cdot 10^{-2}$ & $1.52 \cdot 10^{-2}$ & $3.22 \cdot 10^{3}$ & &
 \\
& 36 
 & $5.17 \cdot 10^{-2}$ & $1.90 \cdot 10^{-3}$ & $3.46 \cdot 10^{3}$ & &
 \\
& 140 
 & $3.80 \cdot 10^{-1}$ & $3.00 \cdot 10^{-4}$ & $3.18 \cdot 10^{3}$ & &
 \\
& 1,224 
 & $4.35 \cdot 10^{1~~}$ & $0.00 \cdot 10^{0~~}$ & $3.04 \cdot 10^{3}$ & &
 \\
\hline
\multirow{3}{*}{\STAB{\rotatebox[origin=c]{90}{Hybrid}}} & 12
 & $4.99 \cdot 10^{-2}$ & $6.27 \cdot 10^{1~~}$ & $1.11 \cdot 10^{3}$ & $1.07 \cdot 10^{4}$ & 7.7\%
 \\
& 36
 & $2.76 \cdot 10^{-2}$ & $1.27 \cdot 10^{2~~}$ & $9.69 \cdot 10^{2}$ & $1.45 \cdot 10^{4}$ & 4.5\%
 \\
& 140
 & $1.83 \cdot 10^{-1}$ & $1.97 \cdot 10^{2~~}$ & $3.25 \cdot 10^{2}$ & $1.84 \cdot 10^{4}$ & 1.2\% 
 \\
& 1,224
 & $1.99 \cdot 10^{1~~}$ & $2.42 \cdot 10^{2~~}$ & $5.30 \cdot 10^{1}$ & $2.10 \cdot 10^{4}$ & 0.028\%
 \\
\hline
\end{tabular}
 \end{center}
\end{table}

\begin{table}[htb]
 \caption{
  Experiments from Table \ref{table:hybrid-vs-comparison-particle-particle} for
  the particle-on-plane setup.
  \label{table:hybrid-vs-comparison:slope}
 }
 \begin{center}
 \footnotesize
 \begin{tabular}{|cr|rrrrr|}
\hline
 & & & & Packed & Packed &
 \\
 & $|\mathbb{T}|$
 & Runtime & Scalar & 128B & 256B & Fallback
 \\
\hline
\multirow{3}{*}{\STAB{\rotatebox[origin=c]{90}{Comp.}}} & 12  
 & $2.60 \cdot 10^{-2}$ & $3.87 \cdot 10^{-2}$ & $3.33 \cdot 10^{3}$ & &
 \\
& 36  
 & $6.50 \cdot 10^{-2}$ & $1.96 \cdot 10^{-2}$ & $3.85 \cdot 10^{3}$ & &
 \\
& 140  
 & $1.84 \cdot 10^{-1}$ & $1.28 \cdot 10^{-2}$ & $4.09 \cdot 10^{3}$ & &
 \\
& 1,224 
 & $1.98 \cdot 10^{0~~}$ & $2.90 \cdot 10^{-3}$ & $4.27 \cdot 10^{3}$ & &
 \\
\hline
\multirow{3}{*}{\STAB{\rotatebox[origin=c]{90}{Hybrid}}} & 12
 & $2.80 \cdot 10^{-2}$ & $1.06 \cdot 10^{1~~}$ & $1.25 \cdot 10^{3}$ & $9.34 \cdot 10^{3}$ & 6.3\%
 \\
& 36
 & $6.70 \cdot 10^{-2}$ & $1.56 \cdot 10^{1~~}$ & $1.34 \cdot 10^{3}$ & $1.19 \cdot 10^{4}$ & 5.1\%
 \\
& 140
 & $1.79 \cdot 10^{-1}$ & $2.15 \cdot 10^{1~~}$ & $1.34 \cdot 10^{3}$ & $1.35 \cdot 10^{4}$ & 4.8\%
 \\
& 1,224
 & $1.85 \cdot 10^{0~~}$ & $3.03 \cdot 10^{1~~}$ & $9.32 \cdot 10^{2}$ & $1.48 \cdot 10^{4}$ & 3.6\%
 \\
\hline
\end{tabular}
 \end{center}
\end{table}

%
%
Our hybrid approach outperforms a sole comparison-based approach robustly for
the \replaced[id=Add]{
$ | \mathbb{T} | : | \mathbb{T} | $
}{
$ | N_{\text{triangles}} | : | N_{\text{triangles}} | $
}
setups.
For strongly ill-balanced triangle counts, the insulated comparison-based
approach is superior (Table~\ref{table:hybrid-vs-comparison-particle-particle}).
The comparison-based code variant is not able to benefit from AVX at all (not
shown), while the hybrid AVX usage increases with increasing
triangle counts.
 We end up with up to 40--45\% ``turbo-mode'' peak performance which we have to
 calibrate with the AVX frequency reduction \cite{Charrier:2019:Energy}.
The relative number of fallbacks,
i.e.~situations where the iterative scheme does not converge within four
iterations, decreases with growing geometry detail, while the same effect is not
as predominant for the particle-on-plane scenario.

%
%
Our data confirm the superiority of the hybrid approach for the
particle-particle comparisons
\cite{Krestenitis:15:FastDEM,Krestenitis:17:FastDEM}.
They confirm that the
approximation of the Hessian does not significantly harm the robustness, even
though the number of fallbacks becomes non-negligible.
Our arithmetic intensity \deleted[id=Add]{dominated
 by the 
 $ \mathcal{O}(| N_{\text{triangles}} |^2) $ for the $ |
 N_{\text{triangles}} | : | N_{\text{triangles}} | $ 
 algorithm in the
 sphere-to-sphere setup as opposed to $ \mathcal{O}(| N_{\text{triangles}} |) $
} determines how much improvement results from the hybrid strategy.
\replaced[id=Add]{
 It is significantly higher for the particle-particle setup as opposed to the
 particle-on-plane.
 Therefore, only the former
}{
 The latter
} prospers through vectorisation.
 We see an increased fallback for decreased geometric detail due to the larger
 relative epsilon, which is used to identify fallback conditions.

%
%
\begin{observation}
 As long as we do not compare extreme cases (single triangle vs.~a lot of
 triangles), the hybrid approach is faster.
 It is thus reasonable to employ it on all levels of the surrogate tree, even
 though it might be reasonable to skip iterative \replaced[id=Add]{comparisons}{comparisions} a priori if the
 coarsest surrogate level is involved.
 The latter observation does not result from a mathematical ``non-robustness''
 but is a sole machine effect.
\end{observation}

\subsection{Multiresolution comparisons for explicit time stepping}
%
%
Within an explicit time stepping code, our multiresolution approach promises to
eliminate unnecessary comparisons since it identifies ``no collission''
constellations quickly through the surrogates:
Whenever it compares two geometries, the algorithm runs through the resolution
levels top-down (from coarse to fine).
The monotonicity of the surrogate definition implies that we can stop
immediately if there is no overlap between two surrogates.
\replaced[id=Add]{The}{Our} code either employs the pure geometry-based approach
or the hybrid strategy on all levels.
\added[id=R1]{
 We focus on spherical particles ($\eta _r$=1) discretised by 1,224
 triangles, and average over 100 time steps such that we 
 run into no-collision phases for the particle-particle setup
 and see the particle roll down the plane for particle-on-plane. 
}

\begin{table}[ht]
 \caption{
  \replaced{Time-to-solution ([t]=s) and number of triangle distance
  checks}{Measurements} for an explicit Euler \replaced[id=R1]{step}{over 100
  time steps} for the particle-particle collision (top) and the
  particle-on-plane setup (bottom).
  We compare a comparision-based setup to a hybrid approach on a single level
  vs.~a surrogate hierarchy which is traversed from coarse to fine.
  For the hybrid configuration, we \replaced[id=Add]{present the number of
  comparison-based fallbacks vs.~the number of iterative comparisons.
  Each iterative comparison of two triangles consists of four Newton
  steps.
  Only if these four steps fail to converge, the algorithm issues the
  comparison-based postprocessing.
  Both the iterative comparisons plus the (fewer) comparison-based
  postprocessing steps determine the runtime (right column).
  }{show both the
  number of iterative sweeps of four iterations and the plain triangle-to-triangle comparisons.}
  \label{table:multiscale:explicit:comparison-count}
 }
\centering
{\footnotesize
\begin{tabular}{|p{1.8cm}|rr|rrr|}
\hline
 & \multicolumn{2}{|c|}{Comparison-based}
 & \multicolumn{3}{|c|}{Hybrid}
 \\ 
Method & \#tri.~comp. & Runtime & \#tri.~comp. & \#iterative & Runtime \\
\hline
Single level 
 & $1.50 \cdot 10^{6}$ & $3.96 \cdot 10^{0~~}$
 & $1.66 \cdot 10^{3}$ & $1.50 \cdot 10^{6}$ & $1.87 \cdot 10^{0~~}$
 \\
Surrogate hierarchy 
 & $8.20 \cdot 10^{3}$ & $2.60 \cdot 10^{-2}$
 & $1.44 \cdot 10^{2}$ & $7.64 \cdot 10^{3}$ & $1.80 \cdot 10^{-2}$ \\
\hline
Single level 
 & $6.27 \cdot 10^{5}$ & $1.68 \cdot 10^{0~~}$
 & $1.57 \cdot 10^{4}$ & $6.27 \cdot 10^{5}$ & $8.00 \cdot 10^{0~~}$\\
Surrogate hierarchy 
 & $5.27 \cdot 10^{3}$ & $7.68 \cdot 10^{-2}$
 & $4.96 \cdot 10^{2}$ & $5.03 \cdot 10^{3}$ & $4.00 \cdot 10^{-2}$\\
\hline
\end{tabular}
}
\end{table}

%
%
Our measurements confirm the superiority of the hybrid scheme in the surrogate
context (Table~\ref{table:multiscale:explicit:comparison-count}):
In line with Section \ref{subsection:hybrid-single-level}, no
multiresolution setup with comparison-based contact detections on surrogate
levels is able to outperform the configurations where all levels are tackled
through the hybrid approach.
Further studies where different variants are used on different (surrogate)
levels are beyond scope.

%
%
In our two-particles scenario, the particles \deleted[id=Add]{each host $1,224$
triangles, and hence} yield $1,224^2$ comparisons per time step if no surrogate
helper data structure is used.
\replaced[id=Add]{
 Our data reflects the quadratic (particle-particle) or linear
 (particle-on-plane) computational complexity.
}{
As we sum up the comparisons over 100 time steps, the $ \mathcal{O}(|
N_{\text{triangles}} |^2) $ complexity delivers exactly the measured total
comparison count.
An analogous argument holds for the sphere-on-plane setup, where the
slope hosts 512 triangles.
}
In both cases, the number of triangle comparisons is reduced by more than an
order of magnitude through the surrogate hierarchy, and the surrogate version
outperforms its single-level baseline robustly.
The hierarchical scheme's additional computational cost (overhead) is
negligible, though it does not significantly alter the ratio of iterative checks
to fallback comparisons in the hybrid scheme.

\begin{observation}
 Our surrogate technique efficiently reduces the number of comparisons between
 two geometries, as \replaced[id=Add]{``no-collision''}{``no-collission''} setups are identified with low
 computational cost.
\end{observation}

\subsection{Multiresolution comparisons for implicit time stepping}
Implicit methods are significantly more stable then their explicit counterpart.
The price to pay for this is an increased computational complexity:
The Picard iterations that we use imply that we have to run the core contact
point detection more often per time step.
The Picard iterations' update of collision point detections imply that the
surrogate tree does not unfold linearly anymore.
While the explicit time stepping algorithm runs through the tree from coarse to
fine, the implicit scheme descends into finer levels yet might, through the
iterative updates of the rotation and position, find alternative tree
parts that have to be taken into account too or instead.
\added[id=R1]{
 The increased stability of implicit time stepping allows users to pick larger
 time step sizes which, in practice, compensate to some degree of the increased
 compute load.
 This problem-specific cost amortisation is beyond scope here, i.e.~we compare
 implicit and explicit schemes with the same time step size.
 All data results from spherical particles and the default triangle counts.
 We average all measurements over 100 time steps and include the ``roll down the
 plane'' situation for the particle-on-plane setup.
}

\begin{table}[ht]
 \caption{ 
  Average number of Picard iterations per time step for our first two scenarios.
  \label{table:multiscale:implicit:Picard-iterations}
 }
\centering
 {\footnotesize
\begin{tabular}{|p{3.2cm}|rr|rr|}
\hline
 & \multicolumn{2}{|c|}{Particle-particle}
 &
 \multicolumn{2}{|c|}{\replaced[id=Add]{Particle-plane}{Sphere-on-plate}}
 \\
Method & Comparison-based & Iterative & Comparison-based & Iterative \\
\hline
Single level or surrogate within Picard& $4.6$ & $4.9$ & $7.1$ & $7.1$ \\
Multiscale Picard & $6.2$ & $6.2$ & $13.0$ & $13.1$\\
\hline
\end{tabular}
}
\end{table}

%
%
\replaced[id=Add]{
 Averaged over 100 time steps,
}{
 Once we study the comparisons over 100 time steps, 
}
we observe that the number
of Picard iterations is small and bounded
(Table~\ref{table:multiscale:implicit:Picard-iterations}).
We study the impact of a switch to the
iterative scheme, with the hybrid fallback on the finest level, and observe that it slightly increases the
Picard iteration count.
The usage of a multiscale method merged into the Picard iterations increases the
iteration count, too.
Both modifications yield flawed contact point
guesses and thus require us to run
more Picard iteration steps overall.
The wrong guesses have to be compensated later on.

\begin{observation}
 Both the iterative approximation of contact points and the ``one Picard step
 before we widen the active set'' strategy increase the total number of required
 Picard iterations.
\end{observation}

\begin{figure}
 \begin{center}
  \includegraphics[width=0.45\textwidth]{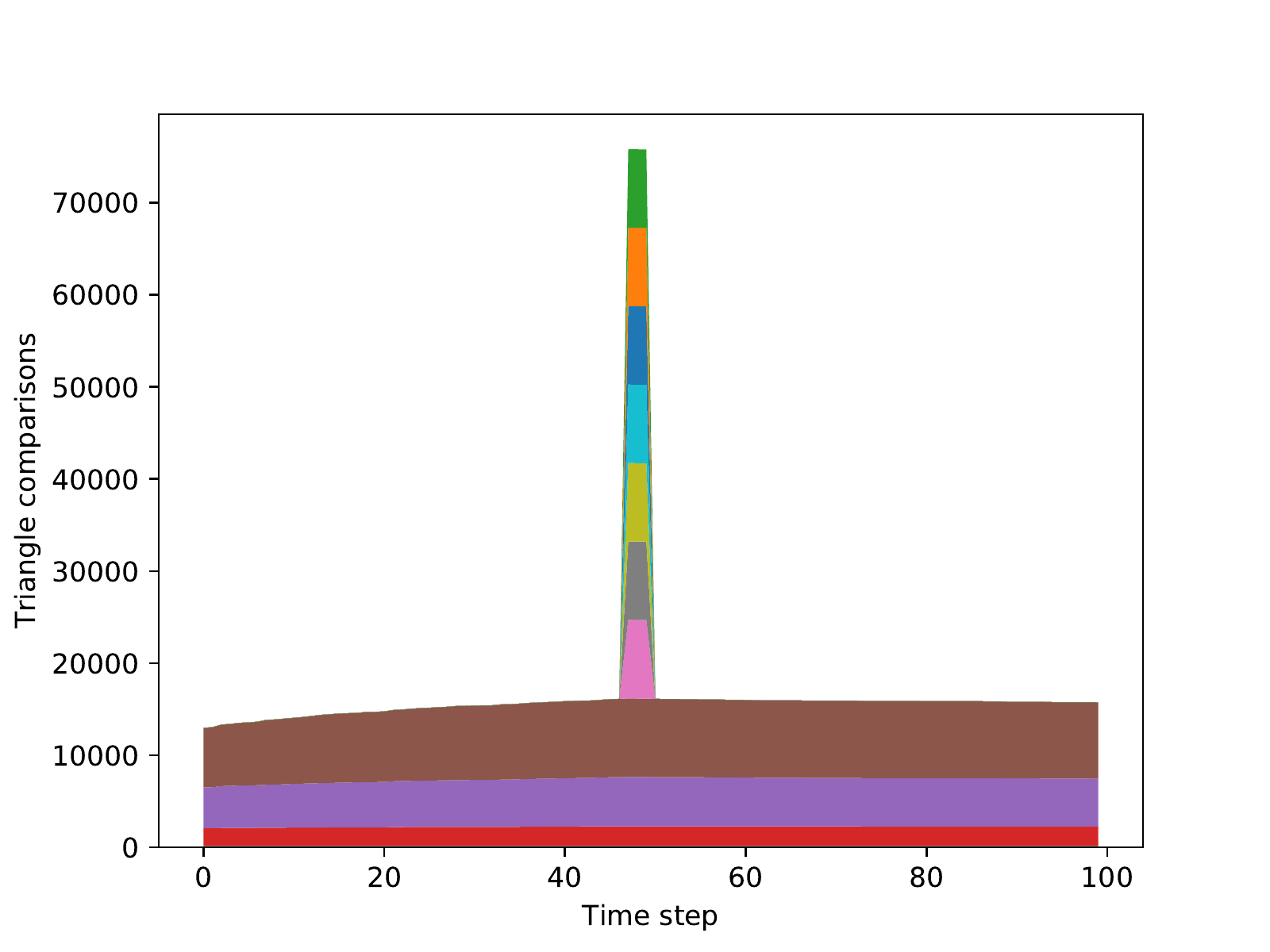}
  \includegraphics[width=0.45\textwidth]{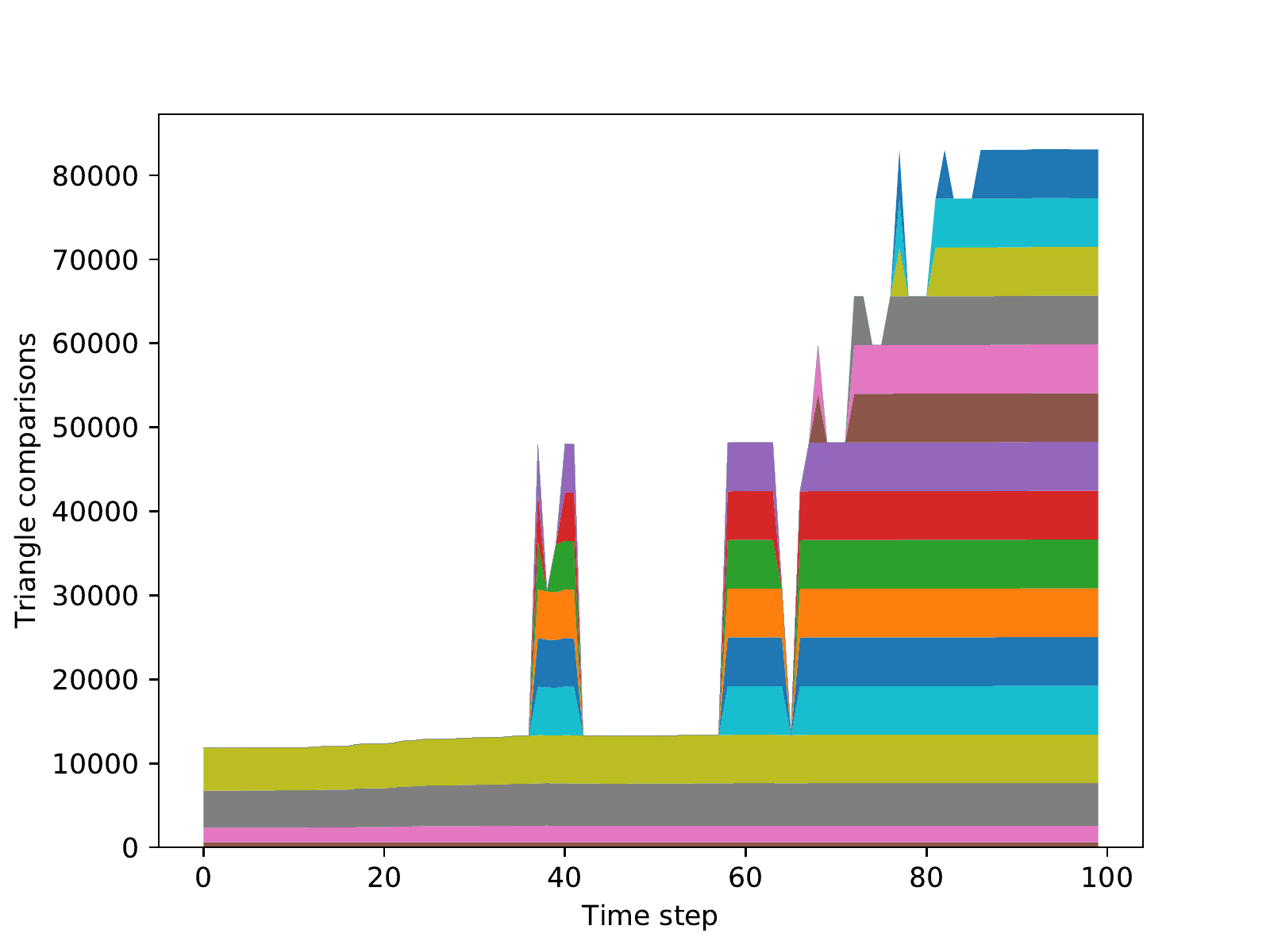}
 \end{center}
 \caption{
  Number of triangle-to-triangle comparisons over time per surrogate
  representation level.
  The data stems from the particle-particle (left) and the
  particle-on-plane setup (right) subject to the implicit time stepping.
  \added[id=R3]{The lowest layer illustrates the triangle comparison count for the first iteration, the next layer for the second iteration and so forth.
  The number of layers corresponds to the total number of iterations.}
  \label{figure:08e_implicit:collision-histogram}
 }
\end{figure}

%
%
\noindent
The surrogate hierarchy yields an efficient early termination criterion for our
collision detection.
If there is no collision, the code does not step down into the fine grid
resolutions.
This property carries over from the explicit to the implicit
algorithm (Figure \ref{figure:08e_implicit:collision-histogram}).
An increase of the computational cost by a factor of 4.6 is acceptable in
return for an implicit scheme.
We however observe that this increase holds for brief point contacts only.
It raises to a factor of 13.1 if contacts persist.
In our example, this happens once the spherical object starts to roll and slide
down the tilted plane.

\begin{table}[ht]
 \caption{
  \replaced[id=R1]{Time per time step ([t]=s])}{Time-to-solution ([t]=s)} and
  triangle distance comparisons for our implicit schemes for the particle-particle collision (top) and  the  particle-on-plane  setup  (bottom).
  \label{table:multiscale:implicit:comparison-count}
 }
\centering
{\footnotesize
\begin{tabular}{|p{1.8cm}|rr|rrr|}
\hline
 & \multicolumn{2}{|c|}{Comparison-based}
 & \multicolumn{3}{|c|}{Hybrid}
 \\ 
Method & \#tri.~comp. & Runtime & \#tri.~comp. & \#iterative & Runtime \\
\hline
Single level
 & $6.89 \cdot  10^{8}$ & $1.69 \cdot 10^{2~~}$
 & $8.78 \cdot 10^{5}$ & $7.34 \cdot 10^{8}$ & $7.44 \cdot 10^{1~~}$\\
Surrogate within Picard
 & $3.97 \cdot 10^{6}$ & $1.08 \cdot 10^{0~~}$
 & $7.14 \cdot 10^{4}$ & $3.70 \cdot 10^{6}$ & $7.20 \cdot 10^{-1}$\\
Multiscale Picard
 & $1.82 \cdot 10^{6}$ & $7.70 \cdot 10^{-1}$
 & $2.25 \cdot 10^{4}$ & $1.70 \cdot 10^{6}$ & $4.70 \cdot 10^{-1}$\\
\hline
Single level 
 & $4.43 \cdot 10^{8}$ & $9.61 \cdot 10^{1~~}$
 & $1.11 \cdot 10^{7}$ & $4.43 \cdot 10^{8}$ & $4.41 \cdot 10^{1~~}$\\
Surrogate hierarchy 
 & $3.59 \cdot 10^{6}$ & $4.80 \cdot 10^{-1}$
 & $3.44 \cdot 10^{5}$ & $3.42 \cdot 10^{6}$ & $2.20 \cdot 10^{-1}$\\
Multiscale Picard
 & $3.50 \cdot 10^{6}$ & $4.90 \cdot 10^{-1}$
 & $3.11 \cdot 10^{5}$ & $3.35 \cdot 10^{6}$ & $1.90 \cdot 10^{-1}$\\
\hline
\end{tabular}
}
\end{table}

%
%
Within our multi-resolution framework, the cost per Picard iteration is not
uniform and constant but depends heavily on the surrogate tree fragments that are used.
The cost are in particular non-uniform for non-simplistic setups.
The growth in Picard iterations (on average) per time step
(Table~\ref{table:multiscale:implicit:Picard-iterations}) increases the number
of triangle-to-triangle checks, compared to the explicit schemes (Table
\ref{table:multiscale:explicit:comparison-count}), by exactly this factor if we
stick to a plain geometry model.
Yet, it does not manifest in an explosion of the
runtime (Table \ref{table:multiscale:implicit:comparison-count}) if we employ
the surrogate trees.
They help to reduce the compute cost dramatically, as we study
only those parts of the surrogate tree which might induce a collision. 
We prune the tree per Picard iteration.

\begin{observation}
 Our multiscale Picard approach is particularly beneficial for strongly
 instationary setups where the topology of particle interactions changes
 quickly.
\end{observation}

%
%
\noindent
Permuting and fusing the Picard iteration loops and the traversal over the
surrogate tree reduces the number of triangle-to-triangle comparisons further.
This observation holds for the particle-particle setup.
It does not hold for the particle-on-plane. 
The advanced version benefits from the fact that we memorise the active set
in-between two Picard iterations:
While the implicit version from Section
\ref{subsection:implicit:multiresolution-acceleration} runs through the whole
tree starting from the root in every iteration, 
our advanced version starts from a certain resolution and
unfolds at most one level per Picard step.
We save the progress through coarser resolutions, and we do not step all the
way down in early iterations.
This state-based approach works as our narrowing is effective:
if we step down into a ``wrong'' part of the tree and find out that 
these fine resolutions do not contribute towards the final force, we 
successively remove these fine resolutions from the (active) comparison 
sets again.
For a sphere rolling or hopping down a plane, the active set remains almost
invariant throughout the Picard iterations, and we do not benefit from the
narrowing or an early termination.
We do however benefit here from the adaptive localisation of the contact
detection within the tree.

\begin{observation}
 The multiscale Picard approach in combination with a hybrid contact detection
 keeps the cost of the implicit time stepping bounded by a factor of four
 compared to an explicit scheme.
\end{observation}

\subsection{Multiresolution comparisons for implicit time stepping with large scale differences}
%
%
\added[id=R2]{
 To assess collisions between objects of different scales, we first
 instantiate the particle-particle scenario with both
 particles being spheres approximated by $|\mathbb{T}|=80$ triangles
 ($\eta _r=1$).
 We fix one of the spheres with a unit diameter, and scale the other one up.
 All data studies situations where the two particles are close and yield contact
 points.
}

\begin{table}[htb]
 \caption{ 
  Time ([t]=s) plus triangle distance comparisons per implicit time step for a 
  particle-particle setup (80 triangles) where one particle is scaled relative
  to the other one.
  \label{table:multiscale:implicit:different-scale-iterations}
 }
\centering
 {\footnotesize
\begin{tabular}{|lr|rrrr|}
\hline
Scale & Triangles & Iterations & \#tri.~comp. & \#iterative & Time \\
\hline
1 & 80 & 13 & $1.59 \cdot 10^{1}$ & $1.45 \cdot 10^{2}$ & $8.72 \cdot 10^{-2}$\\
2 & 80 & 13 & $1.93 \cdot 10^{1}$ & $1.27 \cdot 10^{2}$ & $6.80 \cdot 10^{-2}$\\
4 & 80 & 18 & $4.31 \cdot 10^{1}$ & $1.78 \cdot 10^{2}$ & $1.00 \cdot 10^{-1}$\\
8 & 80 & 14 & $4.10 \cdot 10^{1}$ & $9.82 \cdot 10^{1}$ & $6.82 \cdot 10^{-2}$\\
\hline
\end{tabular}
}
\end{table}

%
%
\added[id=R2]{
 There is no clear trend relating the number of Picard iterations 
 to the scale of the involved particles (Table
 \ref{table:multiscale:implicit:different-scale-iterations}).
 As the scale difference grows, the number of triangle-triangle comparisons
 reduces slightly (cmp.~how often the iterative scheme is invoked), while 
 more iterative updates run into the fallback, i.e.~require a
 posteriori, if-based triangle comparisons.
 The cost to compare a large with a
 small particle thus remains in the same order as the comparison of two equally
 sized particles.
 There is no clear runtime trend.
}

%
%
\added[id=R2]{
 The bigger the size difference between two triangles the more likely it is that
 the hybrid scheme runs into the fallback (cmp.~data in Table
 \ref{table:hybrid-vs-comparison:slope} where the slope triangles are large
 compared to the particle triangles).
 The frequent fallbacks team up with the effect that coarse (surrogate)
 triangles within the larger particle overlap large numbers of
 surrogate triangles on the smaller collision partner.
 Relatively large fractions of the small particle's surrogate tree have to be
 unfolded.
}

\begin{table}[htb]
 \caption{ 
  Time ([t]=s) and triangle distance per time step for two colliding particles
  of different size with comparable triangle sizes.
  \label{table:multiscale:implicit:different-triangles-iterations}
 }
\centering
 {\footnotesize
\begin{tabular}{|lr|rrrr|}
\hline
Scale & Triangles & Iterations & \#tri.~comp & \#iterative & Runtime \\
\hline
1 &    80 & 13 & $1.59 \cdot 10^{1}$ & $1.45 \cdot 10^{2}$ &  $8.72 \cdot 10^{-2}$\\
2 &   320 & 19 & $4.18 \cdot 10^{1}$ & $5.05 \cdot 10^{2}$ & $2.25 \cdot 10^{-1}$\\
4 & 1,280 & 20 & $4.66 \cdot 10^{1}$ & $5.88 \cdot 10^{2}$ & $2.60 \cdot 10^{-1}$\\
8 & 5,120 & 17 & $3.80 \cdot 10^{1}$ & $5.44 \cdot 10^{2}$ & $2.40 \cdot 10^{-1}$\\
\hline
\end{tabular}
}
\end{table}

%
%
\added[id=R2]{
 To distinguish the impact of the particle sizes from the impact of the
 triangle, we next subdivide the upscaled colliding particle such that the
 triangle sizes of both particles remain almost invariant and comparable.
 Doubling the particle size thus corresponds to four times more triangles.
 Within the surrogate tree, we keep the number of triangles per leaf roughly
 constant.
}

%
%
\added[id=R2]{
 Despite the increase in the number of triangles per upscaling, the
 number of comparisons and the runtime quickly plateau (Table
 \ref{table:multiscale:implicit:different-triangles-iterations}).
 The additional cost due to the increase in the triangle count is bounded by a
 factor of four, even if we continue to increase the geometric detail beyond
 a factor of four.
 The cost per triangle reduces as one particle grows relative to the other.
}

\added[id=R2]{
 As the size of the larger object grows, its surrogate tree becomes deeper.
 The algorithm unfolds the larger tree, but this unfolding is a
 localised process.
 We end up with roughly scale-invariant compute cost which is dominated by the
 number of triangles stored within one node of the finest tree level.
}

\begin{observation}
 \added[id=R2]{
 The multiscale Picard approach is particularly effective if we compare objects
 of massively different size yet with the same level of detail.
 In this case, the cost per triangle comparison remains almost invariant.
 }
\end{observation}

\added[id=R2]{
 \noindent
 The observation is particular encouraging for simulations with objects of
 massively different size where the large objects exhibit a high level of
 detail.
 Such situations are found in fluid-structure interaction, where a collision
 induces a force on a bending object.
 Such objects often are modelled via adaptive meshes which have to refine around
 the contact point.
}

\subsection{Many particle experiments}
\added[id=R2]{
 Our multiscale Picard method shows an improvement in time-to-solution for
 individual particles which collide with other particles of different size or
 simple geometric object such as a plane.
 To be useful in real DEM codes, the algorithms must remain performant when we
 run simulations with many particles.
}

\added[id=R2]{
 As long as the neighbourhood search is efficient, our achievements for
 individual particle-particle interactions carry over to dynamic simulations
 with large particle counts.
 If particles move around and bounce into each other, yet each
 particle only interacts with few other objects per time step, our
 multiscale concepts simply are to be
 generalised from two objects to a small object count.
 The situation changes if particles enter a stationary regime, i.e.~if
 the particles are densely clustered and almost at rest.
}

\added[id=R2]{
 We therefore finally simulate 24,576 particles that barely overlap and are
 initially at rest.
 This imitates a basin of granulates, e.g.
 Each particle  is made up of 1,224 triangles, i.e.~we handle
 a total memory footprint of approximately
 3.87GiB which is substantially larger than the 30MiB total L3 cache.
}

\begin{table}[htb]
 \caption{
  \added{Measurements for 24,576 particles which are densely clustered yet
  almost at rest.
  The runtime is the runtime per time step per core. 
  We present the total time and the time per particle-particle comparison.
  }
  \label{table:multiscale:implicit:comparison-count-1024}
 }
\centering
{\footnotesize
\renewcommand{\tabcolsep}{2pt}
    \begin{tabular}{|p{1.6cm}|rrr|rrrr|}
    \hline
     & \multicolumn{3}{|c|}{Comparison-based}
     & \multicolumn{4}{c|}{Hybrid}
     \\ 
    && \multicolumn{2}{c|}{Time} &&& \multicolumn{2}{c|}{Time} \\
    Method & \#tri.~comp. & \multicolumn{1}{c}{total} & \multicolumn{1}{c|}{per
    comp} & \#tri.~comp.
    & \#iterative & \multicolumn{1}{c}{total} & \multicolumn{1}{c|}{per comp} \\
    \hline
    Single level
     & $6.50 \cdot 10^{9}$ & $1.92 \cdot 10^{6~~}$ & $9.67 \cdot 10^{2~~}$
     & $1.29 \cdot 10^{8}$ & $6.50 \cdot 10^{9}$ & $8.36 \cdot 10^{5~~}$ & $4.20 \cdot 10^{2~~}$\\
    Surrogate within Picard
     & $2.49 \cdot 10^{6}$ & $4.70 \cdot 10^{2~~}$ & $2.37 \cdot 10^{-1}$
     & $2.92 \cdot 10^{2}$ & $2.49 \cdot 10^{6}$ & $3.77 \cdot 10^{2~~}$ & $1.90 \cdot 10^{-1}$\\
    Multiscale Picard
     & $9.08 \cdot 10^{5}$ & $2.30 \cdot 10^{2~~}$ & $1.16 \cdot 10^{-1}$
     & $1.09 \cdot 10^{2}$ & $1.06 \cdot 10^{6}$ & $2.13 \cdot 10^{2~~}$ & $1.07 \cdot 10^{-1}$\\
    \hline
    \end{tabular}
}
\end{table}

%
%
\added[id=R2]{
 The hierarchical scheme yields a massive reduction of (iterative)
 triangle-triangle comparisions, and it notably succeeds to avoid too many
 comparisons which subsequently have to be postprocessed (Table
 \ref{table:multiscale:implicit:comparison-count-1024}).
 The runtime per particle-particle comparison is significantly lower than in
 previous dynamic setups.
 Consequently, the cost per particle is lower, too, even though we have around
 six collision partners per particle.
}

%
%
\added[id=R2]{
 Since the particles are almost at rest and only have tiny overlaps, the arising
 forces are very small.
 Consequently, few Picard iterations are sufficient to
 converge for the overall setup.
 The surrogate trees unfold only around the localised contacts, and this
 unfolding is very efficient, i.e.~we barely remove triangles from the active
 set while we run the Picard iterations.
}

\begin{observation}
 \added[id=R2]{
  The multiscale approach remains advantageous for larger particle assemblies
  where the particles are densely clustered. 
 }
\end{observation}

  \section{Conclusion}
\label{section:conclusion}
We present a family of multi-resolution contact detection algorithms that
exhibit low computational cost and high vectorisation efficiency. 
Few core ideas guide the derivation of these algorithms:
We rigorously phrase the underlying mathematics in a multi-resolution and
multi-model language where low-cost resolutions (surrogates) or algorithms
(iterative contact search) precede an expensive follow-up step which becomes
cheaper through good initial guesses or can be skipped in many cases.
We replace dynamic termination criteria behind iterative algorithms with fixed  
iteration counts.
While this might \replaced[id=Add]{mean that we terminate prematurely}{induce
that we terminate prematurel in some cases}, a fixed iteration count allowed us
to unroll loops and to permute them.
The permutation of loops \deleted[id=Add]{finally} is our last ingredient which we apply on
multiple levels: We switch the traversal of triangles with Newton iterations,
and we switch the Picard iterations with the tree unfolding.

The present work is solely algorithmic and has theoretical character.
A natural next step is its application to large-scale, massively
parallel\added[id=R2]{, real-world } simulations.
\added[id=R2]{
 This introduces at least three further challenges:
 First, the parallelisation between multiple compute nodes makes the
 algorithm more complex and introduces significant upscaling and load balancing
 challenges.
 Second, we only use a simple contact model to compute interaction
 forces which might be inappropriate for actual challenges from sciences and
 engineering.
 Finally, our studies solely rely on one, fixed time step size.
 They do not even employ the fact that implicit schemes allow for larger time
 step size; a fact which has to be employed and studied for real-world runs.
}
Furthermore, we rely---so far---on a naive assumption that the Picard iterations 
converge.
A more robust code variant would either identify non-convergence
via force, rotation and movement deltas that do not decrease over
the Picard iterations, or it would exploit the fact that we know how accurate
our surrogate models are via their $\epsilon $ value.
In both cases, surrogate levels could be skipped automatically.

On the methodological side, there are three natural
extensions of our work:
First, our surrogate mechanism always kicks off from the surrogate tree's root
when it searches for contact points.
For time stepping codes, this is not sophisticated. 
It might be advantageous to memorise the tree configurations in-between two
subsequent time steps and thus to exploit the fact that many particle
configurations change only smoothly in time.
\added[id=R2]{Scenarios such as our particle-on-plane setup may particularly
benefit from this, as coarse level surrogates, where the force estimates are
just as likely to push the solution in the wrong direction, can be skipped
based on hints from the previous time step.
}

Second, we work with multiple spatial representations, i.e.~accuracies, but we
stick to single precision all the way through.
It is a natural extension to make our iterative algorithm use a reduced
precision on coarse surrogate models, i.e.~early throughout the algorithm.
Any machine imprecision can be recovered in our case through a slight increase
of $\epsilon $.
Such a mixed precision strategy is particularly
\replaced[id=Add]{attractive}{attrictive} in an era where more and more compute
devices are equipped with special-purpose, reduced-precision linear algebra components.
\added[id=Add]{
 An additional flavour of mixed precision arises once we use more 
 complex contact models.
 It is not clear if such complex models are required on the surrogate levels,
 too, i.e.~it might be appropriate to implement a multi-physics model where
 surrogate contact detection complements accurate interaction models on the
 finest mesh level.
}

Finally, we next will have to tackle large-scale systems implicitly:
DEM models are notoriously stiff, yet the stiffness is localised, as \added[id=Add]{typically} not all
particles in a setup \deleted[id=Add]{typically do} interact with all other particles.
It is a natural extension to investigate \deleted[id=Add]{into} local time stepping where each
particle \replaced[id=Add]{advances}{advanced} with its own $\Delta t$, and to make the surrogate
representations naturally follow and inform these local time step choices.

\section*{Acknowledgements}

This work made use of the facilities of the Hamilton HPC Service of Durham
University.
  
  \appendix
  \section{Surrogate triangles}
\label{appendix:surrogate-triangles}
The construction of good surrogate models is a challenge of its own, as there
are infinitely many surrogate models for a given particle $p$.
We rely a functional minimisation with a fixed coarsening factor to
construct the surrogate hierarchy bottom-up.
Let 
one surrogate triangle for a set of triangles $\mathbb{T}_{k-1}$
be spanned by its three vertices $x_1,x_2,x_3 \in \mathbb{R}^3$ which follow

\begin{eqnarray}
 \argmin_{x_1,x_2,x_3} 
 && 
 \frac{1}{\beta _{\text{size}}}
 \sum _{x \in \{x_1,x_2,x_3\},t \in \mathbb{T}_{k-1}} \| x-t \|^ {\beta
 _{\text{size}}} + 
 \frac{
   \alpha _{\text{area}} 
 }{2}
 |(x_2-x_1) \times (x_3-x_1)|^{-2}
 \nonumber
 \\
 + &&
 \frac{
   \alpha _{\text{inside}}
 }{
   \beta _{\text{normal}} 
 }
 \sum _{x \in \{x_1,x_2,x_3\}, x_t \in t \ \forall t \in \mathbb{T}_{k-1}}
 \left( \max (0, (x-x_t) \cdot N(x_1,x_2,x_3) ) \right)
 ^{\beta _{\text{normal}}}.
 \label{equation:surrogate-triangle}
\end{eqnarray}

\noindent
$\beta _{\text{size}} \geq 2$ is a fixed integer parameter which we
usually pick very high, and the term $\|x-t\|$ denotes the distance between a
point $x$ and a triangle $t$.
The sum thus minimises the maximum distance
between the vertices of the surrogate triangle and the triangles from
$\mathbb{T}_{k-1}$.
We try to make the surrogate triangle as small as possible.
The second term acts as a regulariser that avoids that the surrogate triangle
degenerates and becomes a single point or a line.
Without it, the first term would yield a single point, i.e.~$x_1=x_2=x_3$.

The third term exploits the fact that each triangle $t$ of $p$ has a unique
outer normal $N(t)$.
Even though our surrogate models can be weakly connected, it is thus possible to
assign each surrogate triangle an outer normal, too.
The penalty term over the scalar product
drops out due to the $\max $ function if the
surrogate's normal points into the same direction as the triangles' normals.
Effectively, this term ensures that the surrogate triangle nestles 
closely around a particle and that spikes
do not induce a blown-up surrogate
(Fig.~\ref{figure:multiresolution-model:contact-point}).
Once (\ref{equation:surrogate-triangle}) yields a surrogate triangle,
$\epsilon $ is chosen such that the triangle is conservative for
$\mathbb{T}_{k-1}$.
  \section{Surrogate trees}
\label{appendix:surrogate-trees}
Let $N_{\text{surrogate}}>1$ be the surrogate coarsening factor.
We construct a surrogate tree top-down (Algorithm
\ref{algorithm:construct_surrogate_tree}):

\begin{algorithm}[htb]
 \caption{
   Top-down algorithm to construct a cascade of surrogate models for a given
   triangulation $\mathbb{T}$.
   The algorithm yields a tree defined through $\sqsubseteq _{\text{child}}$ and
   hence allows us to derive a vast set of different, locally adaptive surrogate
   models.
  \label{algorithm:construct_surrogate_tree}  
 }
 {\footnotesize
 \begin{algorithmic}[1]
  \Function{constructSurrogate}{$\mathbb{T}$}
   \State Construct surrogate triangle $t$ for $\mathbb{T}$
    \Comment Solve (\ref{equation:surrogate-triangle})
   \State Assign $t$ smallest $\epsilon $ such that $t^{\epsilon }$ is
    conservative surrogate
   \State Create trivial graph $\mathcal{T}$ with single node $\{ t^{\epsilon }
   \}$ and no edges
   \State \Call{constructSurrogateRecursively}{$t^{\epsilon },\mathbb{T}$}
  \EndFunction
  \Function{constructSurrogateRecursively}
   {$t^{\epsilon }_{\text{local}}$, $\mathbb{T}_{\text{local}}$}
   \If{$\mathbb{T}_{\text{local}} \leq N_{\text{surrogate}} $}
    \State Add node $\mathbb{T}_{\text{local}}$ to $\mathcal{T}$
    \State Add edge $\mathbb{T}_{\text{local}} \sqsubseteq _{\text{child}}
    \{t^{\epsilon }_{\text{local}}\}$ to $\mathcal{T}$
   \Else
   \State Split $\mathbb{T}_{\text{local}}$ into $N_{\text{surrogate}}$ 
    sets
    $\mathbb{T}_{\text{local},0},\mathbb{T}_{\text{local},1},\mathbb{T}_{\text{local},2},\ldots$
    of roughly same size
    \For{$i$}
     \State Construct surrogate triangle $t_{\text{new}}$ for
     $\mathbb{T}_{\text{local},i}$ 
       \Comment Solve (\ref{equation:surrogate-triangle})
     \State Assign $t_{\text{new}}$ smallest $\epsilon $ such that $t^{\epsilon
     }_{\text{new}}$ is conservative surrogate over
     $\mathbb{T}_{\text{local},i}$ 
     \State Add node $\{t^{\epsilon }_{\text{new}}\}$ to $\mathcal{T}$
     \State Add edge $\{t^{\epsilon }_{\text{new}}\} \sqsubseteq _{\text{child}} \{t^{\epsilon }_{\text{local}}\}$ to $\mathcal{T}$
     \State \Call{constructSurrogateRecursively}{$t^{\epsilon
     }_{\text{new}},\mathbb{T}_{\text{local},i}$}
    \EndFor
   \EndIf
  \EndFunction
 \end{algorithmic}
 }
\end{algorithm}

\begin{itemize}
  \item The first triangle that we insert into $\mathcal{T}$ is the coarsest
  surrogate model, i.e.~a degenerated object description consisting of one triangle.  
  \added[id=R1]{
   In line with Definition
   \ref{definition:multiresolution-model:surrogate-tree}, this
  }
  is
  the root node of our surrogate tree.
  \item Recursively dividing creates a tree over sets where all non-leaves have
   cardinality one.
   The leaf sets have a cardinality of roughly $N_{\text{surrogate}}$.
   The number of children per tree node is bounded and typically around
   $N_{\text{surrogate}}$.
  \item We construct the surrogate triangles by copying one triangle out of the
  underlying triangle set. Then, we iteratively minimise the functional 
  (\ref{equation:surrogate-triangle}).
  \item To obtain surrogates with reasonably small $\epsilon $, we
   cluster the triangle sets through a tailored $k$-means algorithm
   \cite{MacQueen1967SomeMF}. The subsets $\mathbb{T}_{\text{local},i}$ thus are
   reasonable compact.
\end{itemize}

\added[id=Add] {
 \noindent
 There are many alternative paradigms to construct surrogate trees: 
 Bounding sphere hierarchies would be an alternative.
 Our approach is relatively slow, yet is exclusively used as pre-processing.  
}

  \ifthenelse{ \boolean{useSISC} }{
    \bibliographystyle{siamplain}
  }{
    \bibliographystyle{plain}
  }
  \bibliography{paper}

\begin{thebibliography}{10}

\bibitem{RealTimeRendering}
{\sc T.~Akenine-Mller, E.~Haines, and N.~Hoffman}, {\em Real-Time Rendering,
  Fourth Edition}, A. K. Peters, Ltd., USA, 4th~ed., 2018.

\bibitem{alonsomarroquin2008efficient}
{\sc F.~Alonso-Marroquin and Y.~Wang}, {\em An efficient algorithm for granular
  dynamics simulation with complex-shaped objects},  (2008),
  \url{https://arxiv.org/abs/0804.0474}.

\bibitem{Alonso:09:EfficientDEM}
{\sc F.~Alonso-Marroquín and Y.~Wang}, {\em An efficient algorithm for
  granular dynamics simulations with complex-shaped objects}, Granular Matter,
  11 (2009), pp.~317--329.

\bibitem{Apinis:2016:WideningNarrowing}
{\sc K.~Apinis, H.~Seidl, and V.~Vojdani}, {\em Enhancing Top-Down Solving with
  Widening and Narrowing}, Springer International Publishing, Cham, 2016,
  pp.~272--288, \url{https://doi.org/10.1007/978-3-319-27810-0_14},
  \url{https://doi.org/10.1007/978-3-319-27810-0_14}.

\bibitem{Baraff97anintroduction}
{\sc D.~Baraff}, {\em An introduction to physically based modeling: Rigid body
  simulation i - unconstrained rigid body dynamics}, in In An Introduction to
  Physically Based Modelling, SIGGRAPH '97 Course Notes, 1997.

\bibitem{Barequet}
{\sc G.~Barequet, B.~Chazelle, L.~J. Guibas, J.~S. Mitchell, and A.~Tal}, {\em
  Boxtree: A hierarchical representation for surfaces in 3d}, Computer Graphics
  Forum, 15 (1996), pp.~387--396,
  \url{https://doi.org/https://doi.org/10.1111/1467-8659.1530387},
  \url{https://onlinelibrary.wiley.com/doi/abs/10.1111/1467-8659.1530387},
  \url{https://arxiv.org/abs/https://onlinelibrary.wiley.com/doi/pdf/10.1111/1467-8659.1530387}.

\bibitem{Charrier:2019:Energy}
{\sc D.~Charrier, B.~Hazelwood, E.~Tutlyaeva, M.~Bader, M.~Dumbser,
  A.~Kudryavtsev, A.~Moskovsky, and T.~Weinzierl}, {\em Studies on the energy
  and deep memory behaviour of a cache-oblivious, task-based hyperbolic pde
  solver}, The International Journal of High Performance Computing
  Applications, 33 (2019), pp.~973--986.

\bibitem{Cundall}
{\sc P.~A. Cundall and O.~D.~L. Strack}, {\em A discrete numerical model for
  granular assemblies}, Geotechnique, 29 (1979), pp.~47--65,
  \url{https://doi.org/10.1680/geot.1979.29.1.47},
  \url{https://doi.org/10.1680/geot.1979.29.1.47},
  \url{https://arxiv.org/abs/https://doi.org/10.1680/geot.1979.29.1.47}.

\bibitem{Dammertz}
{\sc H.~Dammertz, J.~Hanika, and A.~Keller}, {\em Shallow bounding volume
  hierarchies for fast simd ray tracing of incoherent rays}, Comput. Graph.
  Forum, 27 (2008), pp.~1225--1233,
  \url{https://doi.org/10.1111/j.1467-8659.2008.01261.x}.

\bibitem{Eisenacher}
{\sc C.~Eisenacher, G.~Nichols, A.~Selle, and B.~Burley}, {\em Sorted deferred
  shading for production path tracing}, Computer Graphics Forum, 32 (2013),
  \url{https://doi.org/10.1111/cgf.12158}.

\bibitem{EricsonRTCD}
{\sc C.~Ericson}, {\em Real-Time Collision Detection}.

\bibitem{Gottschalk}
{\sc S.~Gottschalk, M.~Lin, and D.~Manocha}, {\em Obbtree: A hierarchical
  structure for rapid interference detection}, Computer Graphics, 30 (1997),
  \url{https://doi.org/10.1145/237170.237244}.

\bibitem{Held1996RealtimeCD}
{\sc M.~Held, J.~Klosowski, and J.~Mitchell}, {\em Real-time collision
  detection for motion simulation within complex environments}, in SIGGRAPH
  '96, 1996.

\bibitem{Iglberger:2010:GranularFlowNonSpherical}
{\sc K.~Iglberger and U.~R{\"{u}}de}, {\em {Massively parallel granular flow
  simulations with non-spherical particles}}, Computer Science - Research and
  Development, 25 (2010), pp.~105--113,
  \url{https://doi.org/10.1007/s00450-010-0114-4}.

\bibitem{Krestenitis:15:FastDEM}
{\sc K.~Krestenitis and T.~Koziara}, {\em Calculating the minimum distance
  between two triangles on simd hardware}, 4 2015.

\bibitem{Krestenitis:17:FastDEM}
{\sc K.~Krestenitis, T.~Weinzierl, and T.~Koziara}, {\em Fast dem collision
  checks on multicore nodes.}, in Parallel processing and applied mathematics :
  12th International conference, PPAM 2017, Lublin, Poland, September 10-13;
  revised selected papers. Part 1., R.~Wyrzykowski, J.~J. Dongarra, E.~Deelman,
  and K.~Karczewski, eds., no.~10777 in Lecture Notes in Computer Science,
  2018, pp.~123--132.

\bibitem{Li:1998:HierarchicalBoundingVolumes}
{\sc T.~Y. Li and J.~S. Chen}, {\em {Incremental 3D collision detection with
  hierarchical data structures}}, Proceedings of the ACM symposium on Virtual
  reality software and technology 1998 - VRST '98, 1998 (1998), pp.~139--144,
  \url{https://doi.org/10.1145/293701.293719}.

\bibitem{MacQueen1967SomeMF}
{\sc J.~B. MacQueen}, {\em Some methods for classification and analysis of
  multivariate observations}, 1967.

\bibitem{McCalpin:1995:Stream}
{\sc J.~D. McCalpin}, {\em Memory bandwidth and machine balance in current high
  performance computers}, IEEE Computer Society Technical Committee on Computer
  Architecture (TCCA) Newsletter,  (1995), pp.~19--25.

\bibitem{Rakotonirina:2018:ConvexShapes}
{\sc A.~D. Rakotonirina and A.~Wachs}, {\em {Grains3D, a flexible DEM approach
  for particles of arbitrary convex shape - Part II: Parallel implementation
  and scalable performance}}, Powder Technology, 324 (2018), pp.~18--35,
  \url{https://doi.org/10.1016/j.powtec.2017.10.033},
  \url{https://doi.org/10.1016/j.powtec.2017.10.033}.

\bibitem{SSETriangle}
{\sc E.~Shellshear and R.~Ytterlid}, {\em Fast distance queries for triangles,
  lines, and points using sse instructions}, Journal of Computer Graphic
  Techniques, 3 (2014), p.~86–110.

\bibitem{Treibig:2010:Likwid}
{\sc J.~Treibig, G.~Hager, and G.~Wellein}, {\em {LIKWID}: {A} {L}ightweight
  {P}erformance-{O}riented {T}ool {S}uite for x86 {M}ulticore {E}nvironments},
  in Proceedings of the 2010 39th International Conference on Parallel
  Processing Workshops, ICPPW '10, IEEE Computer Society, 2010, pp.~207--216.

\bibitem{grain3d}
{\sc A.~Wachs, L.~Girolami, G.~Vinay, and G.~Ferrer}, {\em Grains3d, a flexible
  dem approach for particles of arbitrary convex shape — part i: Numerical
  model and validations}, Powder Technology, 224 (2012), p.~374–389,
  \url{https://doi.org/10.1016/j.powtec.2012.03.023}.

\bibitem{Weinhart2016InfluenceOC}
{\sc T.~Weinhart, C.~Labra, S.~Luding, and J.~Y. Ooi}, {\em Influence of
  coarse-graining parameters on the analysis of dem simulations of silo flow},
  Powder Technology, 293 (2016), pp.~138--148.

\bibitem{Zhao2019APA}
{\sc S.~Zhao and J.~Zhao}, {\em A poly‐superellipsoid‐based approach on
  particle morphology for dem modeling of granular media}, International
  Journal for Numerical and Analytical Methods in Geomechanics,  (2019).

\bibitem{Zhong2016}
{\sc W.~Zhong, A.~Yu, X.~Liu, Z.~Tong, and H.~Zhang}, {\em {DEM/CFD-DEM
  Modelling of Non-spherical Particulate Systems: Theoretical Developments and
  Applications}}, Powder Technology, 302 (2016), pp.~108--152,
  \url{https://doi.org/10.1016/j.powtec.2016.07.010},
  \url{http://dx.doi.org/10.1016/j.powtec.2016.07.010}.

\end{thebibliography}

\end{document}